\documentclass[letter,runningheads,envcountsame,envcountsect]{llncs}
\usepackage{microtype}%
\makeatletter{}%

\usepackage{amsmath,amsthm,amsfonts,amssymb,color,graphicx,float,hyperref}%
\usepackage[ruled,linesnumbered]{algorithm2e}
\newcommand{\flcs}{\FuncSty{FindLCST}}

\newcommand{\wmax}{w_{max}}

\newcommand{\rr}{\mathbb{R}}
\newcommand{\rrplus}{\rr^+}

\renewcommand{\O}{\mathcal{O}}
\newcommand{\Ot}{\tilde{\O}}
\newcommand{\Os}{\O^*}

\newcommand{\ra}{\rightarrow}
\newcommand{\floor}[1]{\left\lfloor #1 \right\rfloor}
\newcommand{\ceil}[1]{\left\lceil #1 \right\rceil}

\newcommand{\prob}[1]{\mathbb{P}\left[ #1 \right]}

\renewcommand{\bold}[1]{\textbf{#1}}
\newcommand{\parabold}[1]{\smallskip \noindent\bold{#1}}

\newcommand{\STC}{\textsf{STC}}
\newcommand{\LSST}{\textsf{LSST}}
\newcommand{\rank}{\mathsf{rank}}
\newcommand{\Detour}{\mathsf{DT}}
\newcommand{\ec}{\mathsf{ec}}
\newcommand{\conge}{\mathsf{cong}}

\newcommand{\NP}{\textsf{NP}}

\newcommand{\poly}{\textsf{poly}}

\newenvironment{pfof}[1]{\begin{proof}[\textbf{Proof of #1: }]}{\end{proof}}

\newcommand{\calA}{\mathcal{A}}

\newcommand{\calC}{\mathcal{C}}

\newcommand{\calG}{\mathcal{G}}

\newcommand{\khalf}{(\floor{k/2} + 1)}

\title{Spanning Tree Congestion and Computation of Generalized Gy\H{o}ri-Lov{\'{a}}sz Partition}
\titlerunning{Spanning Tree Congestion} %

\author{L. Sunil Chandran
	\thanks{This work was done while this author was visiting Max Planck Institute for Informatics, Saarbr\"ucken, Germany,
		supported by Alexander von Humboldt Fellowship.}
	\inst{1} \and
Yun Kuen Cheung
\thanks{Part of the work done while this author was a visitor at the Courant Institute, NYU.
The visit was funded in part by New York University.}
\inst{2} 
\and
Davis Issac \inst{2}}
\institute{Department of Computer Science and Automation, Indian Institute of Science, India.
  \email{sunil@csa.iisc.ernet.in} \and
Max Planck Institute for Informatics, Saarland Informatics Campus, Germany.
  \email{ycheung@mpi-inf.mpg.de~,~dissac@mpi-inf.mpg.de}}

\begin{document}

\maketitle

%
\begin{abstract}
We study a natural problem in graph sparsification, the Spanning Tree Congestion (\STC) problem.
Informally, the \STC~problem seeks a spanning tree with no tree-edge \emph{routing} too many of the original edges.
The root of this problem dates back to at least 30 years ago,
motivated by applications in network design, parallel computing and circuit design.
Variants of the problem have also seen algorithmic applications as a preprocessing step of several important graph algorithms.

For any general connected graph with $n$ vertices and $m$ edges, we show that its STC is at most $\O(\sqrt{mn})$,
which is asymptotically optimal since we also demonstrate graphs with STC at least $\Omega(\sqrt{mn})$.
We present a polynomial-time algorithm which computes a spanning tree with congestion $\O(\sqrt{mn}\cdot \log n)$.
We also present another algorithm for computing a spanning tree with congestion $\O(\sqrt{mn})$;
this algorithm runs in sub-exponential time when $m = \omega(n \log^2 n)$.

For achieving the above results, an important intermediate theorem is \emph{generalized Gy\H{o}ri-Lov{\'{a}}sz theorem},
for which Chen et al.~\cite{CKLRSV2007} gave a non-constructive proof.
We give the first elementary and constructive proof
by providing a local search algorithm with running time $\Os\left( 4^n \right)$,
which is a key ingredient of the above-mentioned sub-exponential time algorithm.
We discuss a few consequences of the theorem concerning graph partitioning, which might be of independent interest.

We also show that for any graph which satisfies certain \emph{expanding properties}, its STC is at most $\O(n)$,
and a corresponding spanning tree can be computed in polynomial time.
We then use this to show that a random graph has STC $\Theta(n)$ with high probability.
\end{abstract}
\makeatletter{}%

\section{Introduction}\label{sect:intro}

\emph{Graph Sparsification/Compression} generally describes a transformation of
a \emph{large} input graph into a \emph{smaller/sparser} graph that preserves certain feature
(e.g., distance, cut, congestion, flow) either exactly or approximately.
The algorithmic value is clear, since the smaller graph might be used as a preprocessed input to an algorithm,
so as to reduce subsequent running time and memory requirement.
In this paper, we study a natural problem in graph sparsification, the Spanning Tree Congestion (\STC) problem.
Informally, the \STC~problem seeks a spanning tree with no tree-edge \emph{routing} too many of the original edges.
The problem is well-motivated by network design applications,
where designers aim to build sparse networks that meet traffic demands, while ensuring no connection (edge) is too congested.
Indeed, the root of this problem dates back to at least 30 years ago under the name of ``load factor''~\cite{BCLR1986,Simonson1987},
with natural motivations from parallel computing and circuit design applications.
The \STC~problem was formally defined by Ostrovskii~\cite{Ostrovskii2004} in 2004,
and since then a number of results have been presented.
The probabilistic version of the \STC~problem, coined as \emph{probabilistic capacity mapping},
also finds applications in several important graph algorithm problems, e.g., the Min-Bisection problem.

Two canonical goals for graph sparsification problems are to understand the trade-off
between the sparsity of the output graph(s) and how well the feature is preserved,
and to devise (efficient) algorithms for computing the sparser graph(s). These are also our goals for the \STC~problem.
We focus on two scenarios: (A) general connected graphs with $n$ vertices and $m$ edges,
and (B) graphs which exhibit certain \emph{expanding properties}:
\begin{itemize}
\item For (A), we show that the spanning tree congestion (STC) is at most $\O(\sqrt{mn})$,
which is a factor of $\Omega(\sqrt{m/n})$ better than the trivial bound of $m$.
We present a polynomial-time algorithm which computes a spanning tree with congestion $\O(\sqrt{mn}\cdot \log n)$.
We also present another algorithm for computing a spanning tree with congestion $\O(\sqrt{mn})$;
this algorithm runs in sub-exponential time when $m = \omega(n \log^2 n)$.
For almost all ranges of average degree $2m/n$, we also demonstrate graphs with STC at least $\Omega(\sqrt{mn})$.
\item For (B), we show that the expanding properties permit us to devise polynomial-time algorithm which computes a spanning tree with congestion $\O(n)$.
Using this result, together with a separate lower-bound argument, we show that a random graph has $\Theta(n)$ STC with high probability.
\end{itemize}

For achieving the results for (A), an important intermediate theorem is \emph{generalized Gy\H{o}ri-Lov{\'{a}}sz theorem},
which was first proved by Chen et al.~\cite{CKLRSV2007}.
Their proof uses advanced techniques in topology and homology theory, and is non-constructive.

\begin{definition}
In a graph $G=(V,E)$, a $k$-connected-partition is a $k$-partition of $V$ into $\cup_{j=1}^k V_j$,
such that for each $j\in [k]$, $G[V_j]$ is connected.
\end{definition}

\begin{theorem}[{\cite[Theorems 25,~26]{CKLRSV2007}}]\label{thm:weighted-Gyori-Lovasz-no-time}
Let $G=(V,E)$ be a $k$-connected~\footnote{For brevity, we say ``$k$-connected'' for ``$k$-vertex-connected'' henceforth.} graph.
Let $w$ be a weight function $w:V\ra \rrplus$. For any $U\subset V$, let $w(U) := \sum_{v\in U} w(v)$.
Given any $k$ distinct terminal vertices $t_1,\cdots,t_k$,
and $k$ positive integers $T_1,\cdots,T_k$ such that for each $j\in [k]$, $T_j\ge w(t_j)$ and $\sum_{i=1}^k T_i = w(V)$,
there exists a $k$-connected-partition of $V$ into $\cup_{j=1}^k V_j$,
such that for each $j\in [k]$, $t_j\in V_j$ and $w(V_j)\le T_j+\max_{v\in V} w(v)-1$.
\end{theorem}

One of our main contributions is to give the first elementary and constructive proof
by providing a \emph{local search} algorithm with running time $\Os\left( 4^n \right)$:\footnote{$\Os$ notation hides all polynomial factors in input size.}

\begin{theorem}\label{thm:weighted-Gyori-Lovasz}
(a)~There is an algorithm which given a $k$-connected graph, computes a $k$-connected-partition satisfying the conditions stated in Theorem~\ref{thm:weighted-Gyori-Lovasz-no-time}
in time $\Os\left( 4^n \right)$.\\
(b)~If we need a $(\floor{k/2}+1)$-partition instead of $k$-partition (the input graph remains assumed to be $k$-connected),
the algorithm's running time improves to $\Os(2^{\O((n/k)\log k)})$.
\end{theorem}

We make three remarks. 
First, the $\Os(2^{\O((n/k)\log k)})$-time algorithm is a key ingredient of our algorithm for computing a spanning tree with congestion $\O(\sqrt{mn})$.
Second, since Theorem \ref{thm:weighted-Gyori-Lovasz-no-time} guarantees the existence of such a partition, the problem of computing
such a partition is not a \emph{decision problem} but a \emph{search problem}.
Our local search algorithm shows that this problem is in the complexity class \textsf{PLS}~\cite{JPY1988};
we raise its completeness in \textsf{PLS} as an open problem.
Third, the running times do not depend on the weights.

\parabold{The \STC~Problem, Related Problems and Our Results.}
Given a connected graph $G=(V,E)$, let $T$ be a spanning tree.
For an edge $e =(u,v) \in E$, its detour with respect to $T$ is the unique path from $u$ to $v$ in $T$;
let $\Detour(e,T)$ denote the set of edges in this detour.
The stretch of $e$ with respect to $T$ is $\left|\Detour(e,T)\right|$, the length of its detour.
The dilation of $T$ is $\max_{e \in E}  \left|\Detour(e,T)\right|$. %
The edge-congestion of an edge $e\in T$ is $\ec(e,T) := \left| \{ f \in E:  e \in \Detour(f,T) \}\right|$, i.e.,
the number of edges in $E$ whose detours contain $e$.
The congestion of $T$ is $\conge(T) := \max_{e\in T} \ec(e,T)$.
The spanning tree congestion (STC) of the graph $G$ is $\STC(G) := \min_T \conge(T)$, where $T$ runs over all spanning trees of $G$.

We note that there is an equivalent cut-based definition for edge-congestion, which we will use in our proofs.
For each tree-edge $e\in T$, removing $e$ from $T$ results in two connected components; let $U_e$ denote one of the components.
Then $\ec(e,T) := \left| E(U_e,V\setminus U_e) \right|$.

Various types of congestion, stretch and dilation problems are studied in computer science and discrete mathematics.
In these problems, one typically seeks a spanning tree  (or some other structure) with minimum congestion or dilation.
We mention some of the well-known problems, where minimization is done over all the spanning trees of the given graph:
\begin {enumerate}%
\item The Low Stretch Spanning Tree (\LSST) problem is to find a spanning tree which minimizes the total stretch of all the edges of $G$.~\cite{AKPW95}
It is easy to see that minimizing the total stretch is equivalent to minimizing the total edge-congestion of the selected spanning tree.
\item The 
\STC~problem is to find a spanning tree of minimum congestion.~\cite{Ostrovskii2004}
\item Tree Spanner Problem is to find a spanning tree of minimum dilation.~\cite{CC1995}
The more general Spanner problem is to find a sparser subgraph of minimum \emph{distortion}.~\cite{ADDJS1993}
\end {enumerate}
There are other congestion and dilation problems which do not seek a spanning tree, but some other structure.
The most famous among them is the Bandwidth problem and the Cutwidth problem; see the survey~\cite{RSV2000} for more details.

Among the problems mentioned above, several strong results were published in connection with the \LSST~problem.
Alon et al.~\cite{AKPW95} had shown a lower bound of $\Omega(\max\{n \log n,m\})$.
Upper bounds have been derived and many efficient algorithms have been devised; the current best upper bound is
$\Ot(m\log n)$.~\cite{AKPW95,EEST2008,ABN2008,KMP2011,AN2012}
Since total stretch is identical to total edge-congestion,
the best upper bound for the \LSST~problem automatically implies an $\Ot(\frac{m}{n}\log n)$ upper bound on the \emph{average} edge-congestion.
But in the \STC~problem, we concern the \emph{maximum} edge-congestion;
as we shall see, for some graphs, the maximum edge-congestion has to be a factor of $\tilde{\Omega}(\sqrt{n^3/m})$ larger than the average edge-congestion.

In comparison, there were not many strong and general results for the \STC~Problem, though it was studied extensively in the past 13 years.
The problem was formally proposed by Ostrovskii~\cite{Ostrovskii2004} in 2004.
Prior to this, Simonson~\cite{Simonson1987} had studied the same parameter under a different name to approximate the cut width of outer-planar graph.
A number of graph-theoretic results were presented on this topic~\cite{Ostrovskii2010,LLO2014,KOY2009,KO2011,BKMO2011}.
Some complexity results were also presented recently~\cite{OOUU2011,BFGOL2012}, but most of these results concern special classes of graphs.
The most general result regarding STC of general graphs is an $\O(n\sqrt{n})$ upper bound by
L\"owenstein, Rautenbach and Regen in 2009~\cite{LRR2009}, and a matching lower bound by Ostrovskii in 2004~\cite{Ostrovskii2004}.
Note that the above upper bound is not interesting when the graph is sparse, since there is also a trivial upper bound of $m$.
In this paper we come up with a strong improvement to these bounds after 8 years: %

\noindent {\bf Theorem (informal): }
For a connected graph $G$ with $n$ vertices and $m$ edges, its spanning tree congestion is at most $\O(\sqrt {mn})$.
In terms of average degree $d_{avg} = 2m/n$, we can state this upper bound as $\O(n \sqrt{d_{avg}})$.
There is a matching lower bound.

Our proof for achieving the $\O(\sqrt {mn})$ upper bound is constructive.
It runs in exponential time in general; for graphs with $m = \omega(n \log^2  n)$ edges, it runs in sub-exponential time.
By using an algorithm of Chen et al.~\cite{CKLRSV2007} for computing \emph{single-commodity confluent flow} from \emph{single-commodity splittable flow},
we improve the running time to polynomial, but with a slightly worse upper bound guarantee of $\O(\sqrt {mn}\cdot \log n)$.

Motivated by an open problem raised by Ostrovskii~\cite{Ostrovskii2011} concerning STC of random graphs,
we formulate a set of \emph{expanding properties}, and prove that for any graph satisfying these properties, its STC is at most $\O(n)$.
We devise a polynomial time algorithm for computing a spanning tree with congestion $\O(n)$ for such graphs.
This result, together with a separate lower-bound argument, permit us to show that for random graph $\calG(n,p)$
with $1 \ge p \ge \frac{c \log n}{n}$ for some small constant $c>1$,\footnote{Note that the \STC~problem is relevant only for connected graphs.
Since the threshold function for graph connectivity is $\frac {\log n}{n}$, this result applies for almost all of the relevant range of values of $p$.}
its STC is $\Theta(n)$ with high probability, thus resolving the open problem raised by Ostrovskii completely.

\parabold{Min-Max Graph Partitioning and the Generalized Gy\H{o}ri-Lov{\'{a}}sz Theorem.}
It looks clear that the powerful Theorem~\ref{thm:weighted-Gyori-Lovasz-no-time} can make an impact on graph partitioning.
We discuss a number of its consequences which might be of wider interest.

Graph partitioning/clustering is a prominent topic in graph theory/algorithms, and has a wide range of applications.%
A popular goal is to partition the vertices into sets such that the number of edges across different sets is \emph{small}.
While the \emph{min-sum} objective, i.e., minimizing the total number of edges across different sets,
is more widely studied, in various applications, the more natural objective is the \emph{min-max} objective,
i.e., minimizing the maximum number of edges leaving each set. The min-max objective is our focus here.

Depending on applications, there are additional constraints on the sets in the partition.
Two natural constraints are (i) balancedness: the sets are (approximately) balanced in sizes, and
(ii) induced-connectivity: each set induces a connected subgraph.
The balancedness constraint appears in the application of \emph{domain decomposition} in parallel computing,
while the induced-connectivity constraint is motivated by divide-and-conquer algorithms for spanning tree construction.
Imposing both constraints simultaneously is not feasible for every graph; for instance,
consider the star graph with more than $6$ vertices and one wants a $3$-partition.
Thus, it is natural to ask, for which graphs do partitions satisfying both constraints exist.
Theorem~\ref{thm:weighted-Gyori-Lovasz-no-time} implies a simple sufficient condition for existence of such partitions.

By setting the weight of each vertex in $G$ to be its degree, and using the elementary fact that the maximum degree
$\Delta(G)\le n \le 2m/k$ for any $k$-connected graph $G$ on $n$ vertices and $m$ edges, we have
\begin{proposition}\label{pr:min-max}
If $G$ is a $k$-connected graph with $m$ edges, then there exists a $k$-connected-partition,
such that the total degree of vertices in each part is at most $4m/k$. Consequently, the min-max objective is also at most $4m/k$.
\end{proposition}
Due to expander graphs, this bound is optimal up to a small constant factor.
This proposition (together with Lemma~\ref{lm:high-connected-STC-full}) implies the following crucial lemma for achieving some of our results.
\begin{lemma}\label{lm:high-connected-STC-short}
Let $G$ be a $k$-connected graph with $m$ edges. Then $\emph{\STC}(G) \le 4m/k$.
\end{lemma}
Proposition~\ref{pr:min-max} can be generalized to include approximate balancedness in terms of number of vertices.
By setting the weight of each vertex to be $cm/n$ plus its degree in $G$, we have
\begin{proposition}\label{pr:min-max-balance}
Given any fixed $c>0$, if $G$ is a $k$-connected graph with $m$ edges and $n$ vertices,
then there exists a $k$-connected-partition such that the total degree of vertices in each part is at most $(2c+4)m/k$,
and the number of vertices in each part is at most $\frac{2c+4}{c}\cdot \frac nk$.
\end{proposition}
\parabold{Further Related Work.}
Concerning \STC~problem, Okamoto et al.~\cite{OOUU2011} gave an $\Os(2^n)$ algorithm for computing the exact STC of a graph.
The probabilistic version of the \STC~problem, coined as \emph{probabilistic capacity mapping},
is an important tool for several graph algorithm problems, e.g., the Min-Bisection problem.
R\"{a}cke~\cite{Racke2008} showed that in the probabilistic setting, distance and capacity are \emph{interchangeable},
which briefly says a general upper bound for one objective implies the same general upper bound for the other.
Thus, due to the above-mentioned results on \LSST, there is an upper bound of $\Ot(\log n)$ on the \emph{maximum average congestion}.
R\"{a}cke's result also implies an $\O(\log n)$ approximation algorithm to the Min-Bisection problem, improving upon
the $\O(\log^{3/2} n)$ approximation algorithm of Feige and Krauthgamer~\cite{FK2002}.
However, in the deterministic setting, such interchanging phenomenon does not hold:
there is a simple tight bound $\Theta(n)$ for dilation, 
but for congestion it can be as high as $\Theta(n\sqrt{n})$.
For the precise definitions, more background and key results about the concepts we have just discussed,
we recommend the writing of Andersen and Feige~\cite{AF2009}.

Graph partitioning/clustering is a prominent research topic with wide applications,
so it comes no surprise that a lot of work has been done on various aspects of the topic;
we refer readers to the two extensive surveys by Schaeffer~\cite{Schaeffer2007} and by Teng~\cite{Teng2016}.
Kiwi, Spielman and Teng~\cite{KST2001} formulated the min-max $k$-partitioning problem
and gave bounds for classes of graphs with \emph{small separators}, which are then improved by Steurer~\cite{Steurer2006}.
On the algorithmic side, many of the related problems are \textsf{NP}-hard,
so the focus is on devising approximation algorithms.
Sparkled by the seminal work of Arora, Rao and Vazirani~\cite{ARV2009} on sparsest cut and of Spielman and Teng~\cite{ST2013} on local clustering,
graph partitioning/clustering algorithms with various constraints have attracted attention across theory and practice;
we refer readers to~\cite{BFKMNNS2014} for a fairly recent account of the development.
The min-sum objective has been extensively studied;
the min-max objective, while striking as the more natural objective in some applications, has received much less attention.
The only algorithmic work on this objective (and its variants) are Svitkina and Tardos~\cite{ST2004} and Bansal et al.~\cite{BFKMNNS2014}.
None of the above work addresses the induced-connectivity constraint.

The classical version of Gy\H{o}ri-Lov{\'{a}}sz Theorem (i.e., the vertex weights are uniform)
was proved independently by Gy\H{o}ri~\cite{Gyori1976} and Lov{\'{a}}sz~\cite{Lovasz1977}.
Lov{\'{a}}sz's proof uses homology theory and is non-constructive.
Gy\H{o}ri's proof is elementary and is constructive implicitly, but he did not analyze the running time.
Polynomial time algorithms for constructing the $k$-partition were devised for $k=2,3$~\cite{STN1990,WK1993},
but no non-trivial finite-time algorithm was known for general graphs with $k\ge 4$.\footnote{In 1994,
there was a paper by Ma and Ma in \emph{Journal of Computer Science and Technology},
which claimed a poly-time algorithm for all $k$. However, according to a recent study~\cite{HL2014}, Ma and Ma's algorithm can fall into an endless loop.
Also, Gy\H{o}ri said the algorithm should be wrong (see~\cite{NRN1997}).}
Recently, Hoyer and Thomas~\cite{HT2016} provided a clean presentation of Gy\H{o}ri's proof by introducing their own terminology,
which we use for our constructive proof of Theorem~\ref{thm:weighted-Gyori-Lovasz-no-time}.

\parabold{Notation.} Given a graph $G=(V,E)$, an edge set $F\subseteq E$ and $2$ disjoint vertex subsets $V_1,V_2\subset V$,
we let $F(V_1,V_2) ~:=~ \left\{~e=\{v_1,v_2\}\in F ~|~ v_1\in V_1~\text{and}~v_2\in V_2~\right\}$.

\section{Technical Overview}\label{sect:technical-overview}

To prove the generalized Gy\H{o}ri-Lov{\'{a}}sz theorem constructively, we follow the same framework of Gy\H{o}ri's proof~\cite{Gyori1976},
and we borrow terminology from the recent presentation by Hoyer and Thomas~\cite{HT2016}.
But it should be emphasized that proving our generalized theorem is not straight-forward,
since in Gy\H{o}ri's proof, at each stage a single vertex is moved from one set to other to make progress,
while making sure that the former set remains connected.
In our setting, in addition to this we also have to ensure that the weights in the partitions do not exceed the specified limit;
and hence any vertex that can be moved from one set to another need \emph{not} be candidate for being transferred.
The proof is presented in Section \ref{sect:gen-GL}.

As discussed, a crucial ingredient for our upper bound results is Lemma~\ref{lm:high-connected-STC-short},
which is a direct corollary of the generalized Gy\H{o}ri-Lov{\'{a}}sz theorem.
The lemma takes care of the highly-connected cases; for other cases we provide a recursive way to construct a low congestion spanning tree;
see Section~\ref{sect:STC-upper} for details.
For showing our lower bound for general graphs, the challenge is to maintain high congestion while keeping density small.
To achieve this, we combine three expander graphs with \emph{little} overlapping between them,
and we further make those overlapped vertices of very high degree.
This will force a tree-edge adjacent to the centroid of any spanning tree to have high congestion;
see Section~\ref{sect:STC-lower} for details.

We formulate a set of expanding properties which permit constructing a spanning tree of better congestion guarantee in polynomial time.
The basic idea is simple: start with a vertex $v$ of high degree as the root.
Now try to grow the tree by keep attaching new vertices to it, while keeping the invariant that
the subtrees rooted at each of the neighbours of $v$ are roughly balanced in size; each such subtree is called a \emph{branch}.
But when trying to grow the tree in a balanced way, 
we will soon realize that as the tree grow, all the remaining vertices may be seen to be adjacent only to a few number of ``heavy'' branches.
To help the balanced growth, the algorithm will identify a \emph{transferable} vertex which is in a heavy branch,
and it and its descendants in the tree can be transferred to a ``lighter'' branch.
Another technique is to use multiple rounds of matching between vertices in the tree and the remaining vertices to attach new vertices to the tree.
This will tend to make sure that all subtrees do not grow uncontrolled.
By showing that random graph satisfies the expanding properties with appropriate parameters,
we show that a random graph has STC of $\Theta(n)$ with high probability.

\section{Generalized Gy\H{o}ri-Lov{\'{a}}sz Theorem}\label{sect:gen-GL}
We prove Theorem~\ref{thm:weighted-Gyori-Lovasz} in this section.
Observe that the classical Gy\H{o}ri-Lov{\'{a}}sz Theorem follows from Theorem~\ref{thm:weighted-Gyori-Lovasz-no-time}
by taking $w(v)=1$ for all $v\in V$ and $T_j=n_j$ for all $j\in[k]$.
We note that a perfect generalization where one requires that $w(V_j)=T_j$ is not possible
--- think when all vertex weights are even integers, while some $T_j$ is odd.

Let $G=(V,E)$ be a $k$-connected graph on $n$ vertices and $m$ edges, and $w:V\ra \rrplus$ be a weight function.
For any subset $U\subseteq V$, $w(U) := \sum_{u\in U} w(u)$. %
Let $\wmax := \max_{v\in V} w(v)$.

\subsection{Key Combinatorial Notions}

We first highlight the key combinatorial notions used for proving Theorem~\ref{thm:weighted-Gyori-Lovasz};
see Figures~\ref{fig:cascade} and~\ref{fig:rank_level} for illustrations of some of these notions.

\parabold{Fitted Partial Partition.}
First, we introduce the notion of \emph{fitted partial partition} (FPP).
An FPP $A$ is a tuple of $k$ subsets of $V$, $(A_1,\dots,A_k)$, such that 
the $k$ subsets are pairwise disjoint, and for each $j\in [k]$:
\begin{enumerate}
	\item $t_j\in A_j$,
	\item $G[A_j]$ is connected and 
	\item $w(A_j)\le T_j+\wmax-1$ (we say the set is \emph{fitted} for satisfying this inequality).
\end{enumerate}
We say an FPP is a \emph{Strict Fitted Partial Partition} (SFPP) if $A_1\cup\dots\cup A_k$ is a proper subset of $V$.
We say the set $A_j$ is \emph{light} if $w(A_j)< T_j$, and we say it is \emph{heavy} otherwise.
Note that there exists at least one light set in any SFPP, for otherwise $w(A_1\cup \dots\cup A_k) \ge \sum_{j=1}^k T_j  = w(V) $, which means $A_1\cup\dots \cup A_k=V$.
Also note that by taking $A_j=\{t_j\}$, we have an FPP, and hence at least one FPP exists.

\parabold{Configuration.}
For a set $A_j$ in an FPP $A$ and a vertex $v\in A_j\setminus \{t_j\}$, we define the \underline{reservoir} of $v$ with respect to $A$,
denoted by $R_A(v)$, as the vertices in the same connected component as $t_j$ in $G[A_j]\setminus \{v\}$. Note that $v\notin R_A(v)$.

For a \emph{heavy} set $A_j$, a sequence of vertices $(z_1,\dots,z_p)$ for some $p\ge 0$ is called a \underline{cascade} of $A_j$
if $z_1\in A_j\setminus\ \{t_j\}$ and $z_{i+1}\in A_j\setminus R_A(z_i)$ for all $1\le i<p$.
The cascade is called a \underline{null cascade} if $p=0$, i.e., if the cascade is empty.
Note that for light set, we do \emph{not} need to define  its cascade since we do not use it in the proof. (See Figure~\ref{fig:cascade}.)

\begin{figure}[h]
\centering
\includegraphics[scale=0.4]{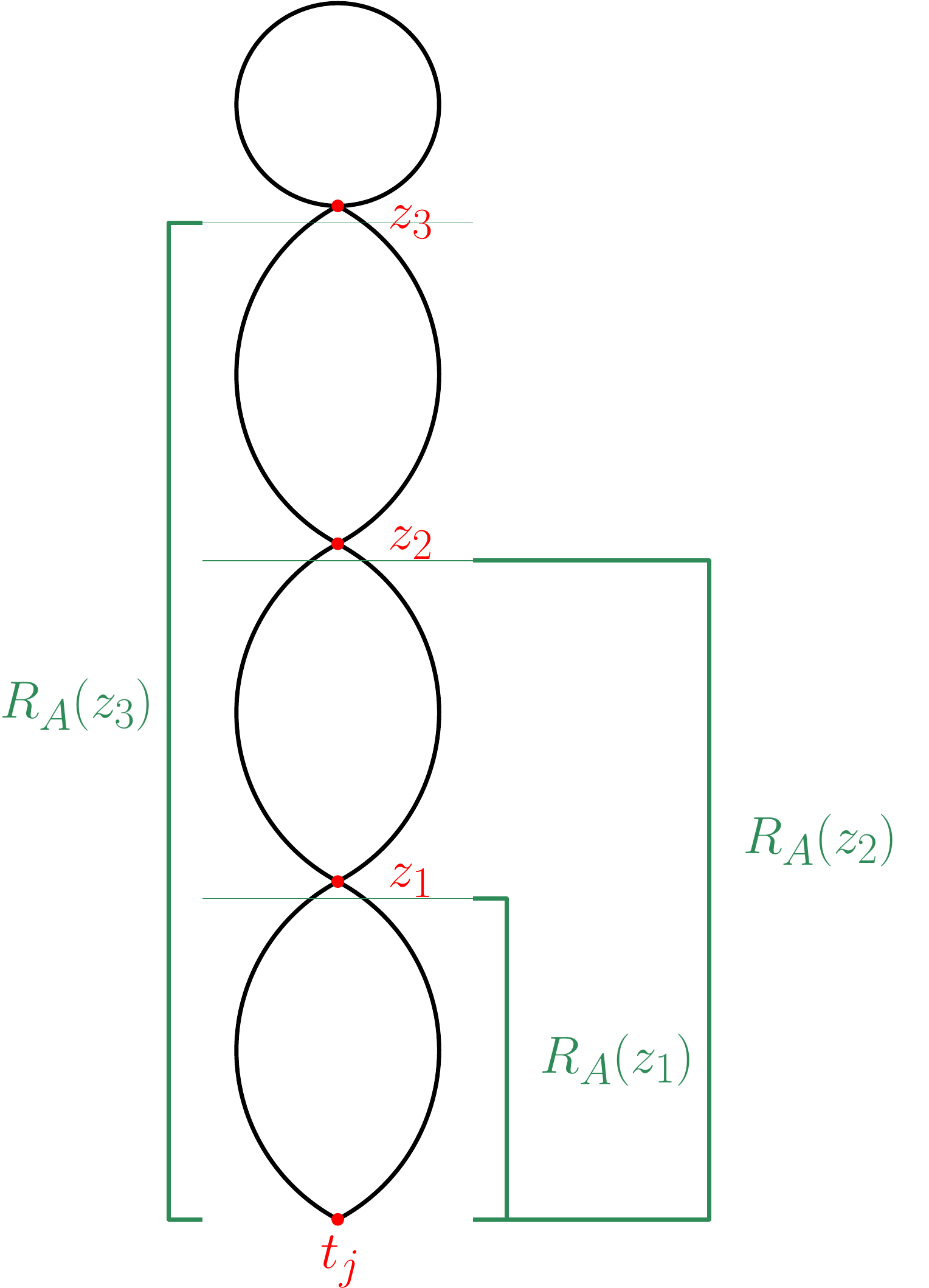}
\caption{Given a configuration $(A,D)$ and a heavy set $A_j$ in $A$,
the figure shows a cascade $(z_1,z_2,z_3)$ for the heavy set $A_j$ and several reservoirs of the cascade vertices.\newline
For any $z_\ell$, $z_\ell \notin R_A(z_\ell)$.
A cascade vertex $z_\ell$ is a cut-vertex of $G[A_j]$, i.e., $G[A_j\setminus \{z_\ell\}]$ is disconnected.
The removal of $z_\ell$ from $A_j$ will lead to at least two connected components in $G[A_j\setminus \{z_\ell\}]$,
and the connected component containing $t_j$ is the reservoir of $z_\ell$.\newline
We identify $t_j = z_0$, but we clarify that a terminal vertex is never in a cascade.
Each epoch between $z_\ell$ and $z_{\ell+1}$,
and also the epoch above $z_3$, is a subset of vertices $B\subset A_j$, where $B\ni z_\ell$ and $G[B]$ is connected.
Note that in general, it is possible that there is no vertex above the last cascade vertex.}
\label{fig:cascade}
\end{figure}

A \underline{configuration} $\calC_A$ is defined as a pair $(A,D)$,
where $A=(A_1,\cdots,A_k)$ is an FPP,
and $D$ is a set of cascades, which consists of exactly one cascade (possibly, a null cascade) for each heavy set in $A$.
A vertex that is in some cascade of the configuration is called a \underline{cascade vertex}.

Given a configuration, we define \underline{rank} and \underline{level} inductively as follows.
Any vertex in a light set is said to have level $0$.
For $i\ge 0$, a cascade vertex is said to have rank $i+1$ if it has an edge to a level-$i$ vertex
but does not have an edge to any level-$i'$ vertex for $i'<i$.
A vertex $u$ is said to have level $i$, for $i\ge 1$, if $u\in R_A(v)$ for some rank-$i$ cascade vertex $v$,
but $u\notin R_A(w)$ for any cascade vertex $w$ such that rank of $w$ is less than $i$.
A vertex that is not in $R_A(v)$ for any cascade vertex $v$ is said to have level $\infty$.

A configuration is called a \underline{valid configuration} if for each heavy set $A_j$,
rank is defined for each of its cascade vertices and the rank is strictly increasing in the cascade,
i.e., if $\left\{ z_1,\dots,z_p \right\}$ is the cascade, then $\rank(z_1)<\cdots<\rank(z_p)$.
Note that by taking $A_j=\{t_j\}$ and taking the null cascade for each heavy set
(in this case $A_j$ is heavy if $w(t_j)=T_j$), we get a valid configuration. (See Figure~\ref{fig:rank_level}.)

\begin{figure}[h]
\centering
\includegraphics[scale=0.48]{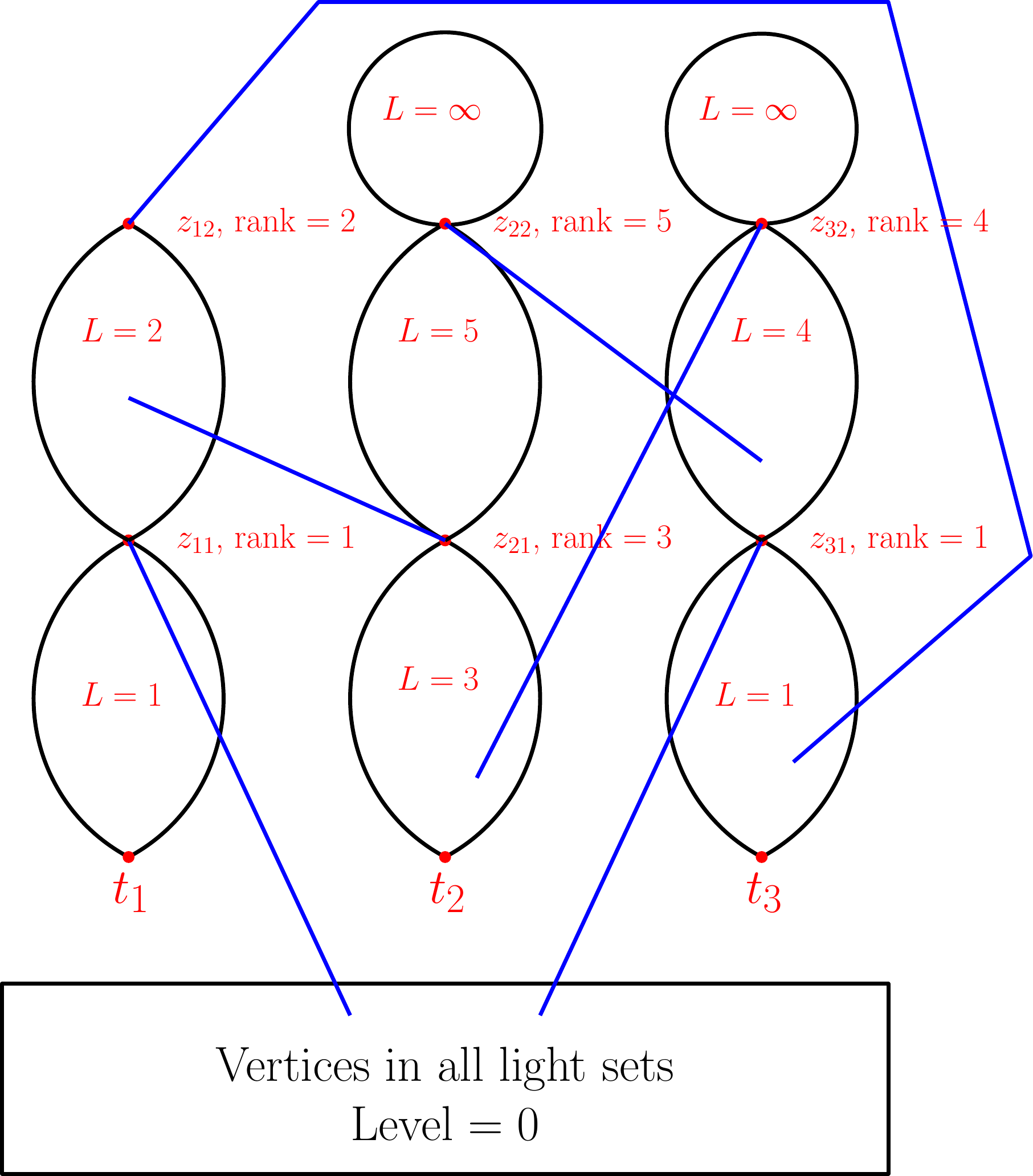}
\caption{An instance of a valid configuration. Every blue segment/curve represent an edge from a cascade vertex to a vertex in some reservoir or light set.\newline
Every cascade vertex connected to a light set has rank $1$, and all vertices in the epoch immediately below a rank $1$ cascade vertex
are of level $1$. Inductively, every cascade vertex connected to a vertex of level $i$ has rank $i+1$,
and all vertices in the epoch immediately below a rank $i$ cascade vertex
are of level $i$.
All vertices above the last cascade vertex of each cascade has level $\infty$.}
\label{fig:rank_level}
\end{figure}

\parabold{Configuration Vectors and Their Total Ordering.}
For any vertex, we define its neighborhood level as the smallest level of any vertex adjacent to it.
A vertex $v$ of level $\ell$ is said to satisfy \underline{maximality property}
if each vertex adjacent on it is either a rank-$(\ell+1)$ cascade vertex, has a level of at most $\ell+1$, or is one of the terminals $t_j$ for some $j$.
For any $\ell\ge 0$, a valid configuration is called an \underline{$\ell$-maximal configuration}
if all vertices having level at most $\ell-1$ satisfy the maximality property.
Note that by definition, any valid configuration is a $0$-maximal configuration.

For a configuration $\calC_A=\left( \left( A_1,\dots,A_k \right),D \right)$, we define $S_A:=V\setminus (A_1\cup\dots\cup A_k)$.
An edge $uv$ is said to be a \underline{bridge} in $\calC_A$ if $u\in S_A$, $v\in A_j$ for some $j\in[k]$, and $level(v)\neq\infty$.

A valid configuration $\calC_A$ is said to be \underline{$\ell$-good} if the highest rank of a cascade vertex in $\calC_A$ is exactly $\ell$
(if there are no cascade vertices, then we take the highest rank as $0$), $\calC_A$ is $\ell$-maximal,
and all bridges $uv$ in $\calC_A$ (if any) are such that $u\in S_A$ and $level(v)=\ell$.
Note that taking  $A_j=\left\{ t_j \right\}$ and taking the null cascade for each heavy set gives a $0$-good configuration.

For each configuration $\calC_A=(A,D)$, we define a \underline{configuration vector} as below:
$$
(~L_A~,~N^0_A~,~N^1_A~,~N^2_A~,~\dots~,~N^n_A~),
$$
where $L_A$ is the number of light sets in $A$, and $N^\ell_A$ is the total number of all level-$\ell$ vertices in $\calC_A$.

Next, we define ordering on configuration vectors.
Let $\calC_A$ and $\calC_B$ be configurations.
We say $\calC_A>_0 \calC_B$ if
\begin{itemize}
	\item $L_A < L_B$, or
	\item $L_A=L_B$, and $N^{0}_A > N^{0}_B$.
\end{itemize}
We say $\calC_A=_0\calC_B$ if $L_A=L_B$ and $N^{0}_A=N^0_B$.
We say $\calC_A\ge_0\calC_B$ if $\calC_A=_0\calC_B$ or $\calC_A>_0\calC_B$.
We say $\calC_A=_\ell\calC_B$ if $L_A=L_B$, and $N^{\ell'}_A=N^{\ell'}_B$ for all $\ell'\le \ell$.

For $1\le \ell \le n$, we say $\calC_A>_\ell \calC_B$ if
\begin{itemize}
	\item $\calC_A >_{\ell-1} \calC_B$, or
	\item $\calC_A=_{\ell-1}\calC_B$, and $N^{\ell}_A >N^{\ell}_B$.
\end{itemize}
We say $\calC_A\ge_\ell \calC_B$ if $\calC_A=_\ell \calC_B$ or $\calC_A>_\ell \calC_B$.
We say $\calC_A>\calC_B$ ($\calC_A$ is \emph{strictly better} than $\calC_B$) if $\calC_A>_n \calC_B$.

\subsection{Proof of Theorem \ref{thm:weighted-Gyori-Lovasz}}

We use two technical lemmas about configuration vectors and their orderings to prove Theorem \ref{thm:weighted-Gyori-Lovasz}(a).
The proof of Theorem \ref{thm:weighted-Gyori-Lovasz}(b) follows closely with the proof of Theorem \ref{thm:weighted-Gyori-Lovasz}(a),
but makes use of an observation about the rank of a vertex in the local search algorithm, to give an improved bound on the number of configuration vectors navigated by the algorithm.
\begin{lemma}
	\label{lem:good-inc}
	Given any $\ell$-good configuration $\calC_A=\left( A=(A_1,\dots,A_k),D_A \right))$ that does not have a bridge, we can find an $(\ell+1)$-good configuration $\calC_B=\left(B = \left( B_1,B_2,\dots,B_k \right),D_B \right)$ in polynomial time such that $\calC_B> \calC_A$.
\end{lemma}
\begin{proof}
Since $\calC_A$ is $\ell$-maximal, any vertex that is at level $\ell'<\ell$ satisfies maximality property.
So, for satisfying $(\ell+1)$-maximality, we only need to worry about the vertices that are at level $\ell$.
Let $X_j$ be the set of all vertices $x\in A_j$ such that $x$ is adjacent to a level-$\ell$ vertex, $level(x)\ge\ell+1$ (i.e., $level(x)=\infty$ as the highest rank of any cascade vertex is $\ell$), $x\neq t_j$, and $x$ is not a cascade vertex of rank $\ell$.

We claim that there exists at least one $j$ for which $X_j$ is not empty. 
If that is not the case, then  we exhibit a cut set of size at most $k-1$.
For each $j$ such that $A_j$ is a heavy set with a non-null cascade, let $y_j$ be the highest ranked cascade vertex in $A_j$.
For each $j$ such that $A_j$ is a heavy set with a null cascade, let $y_j$ be $t_j$.
Let $Y$ be the set of all $y_j$ such that $A_j$ is a heavy set.
Note that $|Y|\le k-1$ as $A$ is an SFPP and hence has at least one light set.
Let $Z_{\infty}$ be the set of all vertices in $V\setminus Y$ that have level $\infty$ and $Z$ be the remaining vertices in $V\setminus Y$.
Since $A$ is an SFPP, $S_A\neq\emptyset$, and since all vertices in $S_A$ have level $\infty$,  we have that $Z_\infty\neq \emptyset$.
$Z$ is not empty because there exists at least one light set in $A$ and the vertices in a light set have level $0$.
We show that there is no edge between $Z_\infty$ and $Z$ in $G$.
Suppose there exists an edge $uv$ such that $u\in Z_\infty$ and $v\in Z$.
If $u\in S_A$, then $uv$ is a bridge which is a contradiction by our assumption that $\calC_A$ does not have a bridge.
Hence $u\in A_j$ for some $j\in[k]$.
Note that $A_j$ has to be a heavy set, otherwise $u$ has level $0$.
We have that $u$ is not a cascade vertex (as all cascade vertices with level $\infty$ are in $Y$) and $u\neq t_j$ (as all $t_j$ such that $level(t_j)=\infty$ are in $Y$).
Also, $v$ is not of level $\ell$ as otherwise, $u\in X_j$ but we assumed $X_j$ is empty. 
But then, $v$ has level at most $\ell-1$, $u$ has level $\infty$, and there is an edge $uv$.
This means that $\calC_A$ was not $\ell$-maximal, which is a contradiction.
Thus, there exists at least one $j$ for which $X_j$ is not empty.

For any $j$ such that $X_j\neq \emptyset$ , there is at least one vertex $x_j$ such that $X_j\setminus \{x_j\}\subseteq R_A(x_j)$.
Now we give the configuration $\calC_B$ as follows.	
We set $B_{j}=A_{j}$ for all $j\in[k]$. 
For each heavy set $A_j$ such that $X_j\neq \emptyset$, we take the cascade of $B_j$ as the cascade of $A_j$ appended with $x_j$.
For each heavy set $A_j$ such that $X_j= \emptyset$, we take the cascade of $B_j$ as the cascade of $A_j$.
It is easy to see that $\calC_B$ is $(\ell+1)$-maximal as each vertex that had an edge to level-$\ell$ vertices in $\calC_A$ is now either a rank $\ell+1$ cascade vertex or a level-$(\ell+1)$ vertex or is $t_j$ for some $j$.
Also, notice that all the new cascade vertices that we introduce (i.e., the $x_j$'s) have their rank as $\ell+1$ and there is at least one rank $\ell+1$ cascade vertex as $X_j$ is not empty for some $j$.
Since there were no bridges in $\calC_A$, all bridges in $\calC_B$ has to be from $S_B$ to a vertex having level $\ell+1$.
Hence, $\calC_B$ is $(\ell+1)$-good.
All vertices that had level at most $\ell$ in $\calC_A$ retained their levels in $\calC_B$. And, at least one level-$\infty$ vertex of 
$\calC_A$ became a level-$(\ell+1)$ vertex in $\calC_B$ because the cascade vertex that was at rank $\ell$ becomes level-$(\ell+1)$ vertex now in at least one set. 
Since $\calC_A$ had no level-$(\ell+1)$ vertices, this means that $\calC_B>\calC_A$.
\end{proof}

\begin{lemma}	\label{lem:bridging}
Given an $\ell$-good configuration $\calC_A=(A=(A_1,\dots,A_k),D_A)$ having a bridge, we can find in polynomial time a valid configuration $
\calC_B=\left( B=\left( B_1,\dots,B_k \right),D_B \right)$  such that one of the following holds:
\begin{itemize}
	\item $\calC_B>_\ell \calC_A$, and $\calC_B$ is an $\ell$-good configuration, or
	\item $\calC_B\ge_{\ell-1} \calC_A$, there is a bridge $u'v'$ in $\calC_B$ such that $u'\in S_B$ and $level(v')\le\ell-1$, and  $\calC_B$ is an $(\ell-1)$-good configuration.
\end{itemize}
\end{lemma}
\begin{proof}
Let $uv$ be a bridge where $u\in S_A$.
Let $A_{j^*}$ be the set containing $v$. 
Note that $level(v)=\ell$ because $\calC_A$ is $\ell$-good.
We keep $B_j=A_j$ for all $j\neq j^*$. 
But we modify $A_{j^*}$ to get $B_{j^*}$ as described below. 
We maintain that if $A_{j}$ is a heavy set then $B_j$ is also a heavy set for all $j$, and hence maintain that $L_B\le L_A$. 

\emph{Case 1: $A_{j^*}$ is a light set} (i.e., when $\ell=0$).
We take $B_{j^*}=A_{j^*}\cup \left\{ u \right\}$.  
For all $j$ such that $B_j$ is a heavy set, cascade of $B_j$ is taken as the null cascade. 
We have $w(A_{j^*})\le T_j-1$ because $A_{j^*}$ is a light set.
So, $w(B_{j^*})=w(A_{j^*})+w(u)\le (T_j-1)+\wmax$, and hence $B_{j^*}$ is fitted.
Also, $G[B_{j^*}]$ is connected and hence $(B_1,\dots, B_k)$ is an FPP\@.
We have $\calC_B>_0\calC_A$ because either $B_{j^*}$ became a heavy set in which case $L_B<L_A$, or it is a light set in which case $L_B=L_A$ and $N^0_B>N^0_A$.
It is easy to see that $\calC_B$ is $0$-good.

\emph{Case 2: $A_{j^*}$ is a heavy set i.e., when $\ell\ge 1$}. 

\emph{Case 2.1: $w(A_{j^*}\cup \{u\})\le T_j+\wmax-1$}.
We take $B_{j^*}=A_{j^*}\cup \{u\}$.
For each $j$ such that $B_{j}$ is a heavy set ($A_j$ is also heavy set for such $j$), the cascade of $B_{j}$ is taken as the cascade of $A_j$. 
$G[B_{j^*}]$ is clearly connected and $B_{j^*}$ is fitted by assumption of the case that we are in.
Hence $B$ is indeed an FPP.
Observe that all vertices that had level $\ell'\le \ell$ in $\calC_A$ still has level $\ell'$ in $\calC_B$. 
Since $level(v)$ was $\ell$ in $\calC_A$ by $\ell$-goodness of $\calC_A$, $u$ also has level $\ell$ in $\calC_B$; 
and $u$ had level $\infty$ in $\calC_A$.
Hence, $\calC_B>_\ell \calC_A$.
It is also easy to see that $\calC_B$ remains $\ell$-good.

\emph{Case 2.2: $w(A_{j^*}\cup \{u\})\ge T_j+\wmax$.}
Let $z$ be the cascade vertex of rank $\ell$ in $A_{j^*}$.
Note that $A_{j^*}$ should have such a cascade vertex as $v\in A_{j^*}$ has level $\ell$.
Let $\bar{R}$ be $A_{j^*}\setminus (R_A(z)\cup z)$, i.e., $\bar{R}$ is the set of all vertices in $A_{j^*} \setminus \{z\}$ with level $\infty$.
We initialize $B_{j^*}:=A_{j^*}\cup \{u\}$.
Now, we delete vertices one by one from $B_{j^*}$ in a specific order until $B_{j^*}$ becomes fitted.
We choose the order of deleting vertices such that $G[B_{j^*}]$ remains connected.
Consider a spanning tree $\tau$ of $G[\bar{R}\cup \{z\}]$.
$\tau$ has at least one leaf, which is not $z$.
We delete this leaf from $B_{j^*}$ and $\tau$.
We repeat this process until $\tau$ is just the single vertex $z$ or $B_{j^*}$ becomes fitted.
If $B_{j^*}$ is not fitted even when $\tau$ is the single vertex $z$, then delete $z$ from $B_{j^*}$.
If $B_{j^*}$ is still not fitted then delete $u$ from $B_{j^*}$.
Note that at this point $B_{j^*}\subset A_{j^*}$  and hence is fitted.
Also, note that $G[B_{j^*}]$ remains connected. 
Hence $\left( B_1,\dots,B_k \right)$ is an FPP.
$B_{j^*}$ does not become a light set because  $B_j$ became fitted when the last vertex was deleted from it. 
Before this vertex was deleted, it was not fitted and hence had weight at least $T_{j^*}+\wmax$ before this deletion.
Since the last vertex deleted has weight at most $\wmax$, $B_{j^*}$ has weight at least $T_{j^*}$ and hence is a heavy set.
Now we branch into two subcases for defining the cascades.

\emph{Case 2.2.1: $z\in B_{j*}$ (i.e, $z$ was not deleted from $B_{j^*}$ in the process above).} For each $j$ such that $B_{j}$ is a heavy set, the cascade of $B_{j}$ is taken as the cascade of $A_{j}$.
Since a new $\ell$ level vertex $u$ is added and all vertices that had level at most $\ell$ retain their level, we have that
$\calC_B>_\ell \calC_A$.
It is also easy to see that $\calC_B$ remains $\ell$-good.

\emph{Case 2.2.2: $z\notin B_{j^*}$ (i.e, $z$ was deleted from $B_{j^*}$).} 
For each $j$ such that $B_{j}$ is a heavy set, the cascade of $B_{j}$ is taken as the cascade of $A_{j}$ but with the rank $\ell$ cascade vertex (if it has any) deleted from it.
$\calC_B\ge_{\ell-1} \calC_A$ because all vertices that were at a level of $\ell'=\ell - 1$ or smaller, retain their levels. 
Observe that there are no bridges in $\calC_B$ to vertices that are at a level at most $\ell-2$, all vertices at a level at most $\ell-2$ still maintain the maximality property, and we did not introduce any cascade vertices.
Hence, $\calC_B$ is $(\ell-1)$-good.
It only remains to prove that there is a bridge $u'v'$ in $\calC_B$ such that $level(v')\le\ell-1$.
We know $z\in S_B$.
Since $z$ was a rank $\ell$ cascade vertex in $\calC_A$, $z$ had an edge to $z'$ such that $z'$ had level $\ell-1$ in $\calC_A$.
Observe that level of $z'$ is at most $\ell-1$ in $\calC_B$ as well.
Hence, taking $u'v'=zz'$ completes the proof.
\end{proof}

\begin{pfof}{Theorem~\ref{thm:weighted-Gyori-Lovasz}(a)}
We always maintain a configuration $\calC_A=(A,D_A)$ that is $\ell$-good for some $\ell\ge 0$.
If the FPP $A$ is not an SFPP at any point, then we are done.
So assume $A$ is an SFPP.

We start with the $0$-good configuration where $A_j=\left\{ t_j \right\}$ and the cascades of all heavy sets are null cascades.
If our current configuration $\calC_A$ is an $\ell$-good configuration that has no bridge,
then we use Lemma~\ref{lem:good-inc} to get a configuration $\calC_B$ such that $\calC_B>\calC_A$ and $B$ is $(\ell+1)$-good.
We take $\calC_B$ as the new current configuration $\calC_A$.
If our current configuration $\calC_A$ is an $\ell$-good configuration with a bridge, then 
we get an $\ell'$-good configuration $\calC_B$ for some $\ell'\ge 0$ such that $\calC_{B}>\calC_A$
by repeatedly applying Lemma~\ref{lem:bridging} at most $\ell$ times.
So in either case, we get a strictly better configuration that is $\ell'$-good for some $\ell'\ge 0$ in polynomial time.
We call this an iteration of our algorithm.

Notice that the number of iterations possible is at most the number of distinct configuration vectors possible.
It is easy to see that the number of distinct configuration vectors with highest rank at most $r$ is at most $ {n+r-1 \choose n}$.
Since rank of any point is at most $n$, the number of iterations of our algorithm is at most $(k+1)\cdot {2n\choose n}$, which is at most $n\cdot 4^n$.
Since each iteration runs in polynomial time as guaranteed by the two lemmas, the required running time is $\Os(4^n)$.

When the algorithm terminates, the FPP given by the current configuration is not an SFPP and this gives the required partition. 
\end{pfof}

\begin{pfof}{Theorem~\ref{thm:weighted-Gyori-Lovasz}(b)}
Since any $k$-connected graph is also $\khalf-$vertex connected, the algorithm will give the required partition
due to Theorem~\ref{thm:weighted-Gyori-Lovasz}(a).
We only need to prove the better running time claimed by Theorem~\ref{thm:weighted-Gyori-Lovasz}(b).
For this, we show that the highest rank attained by any vertex during the algorithm is at most $2n/(k-2)$.
Since the number of distinct configuration vectors with highest rank
$r$ is at most $ {n+r-1 \choose n}$, we then have that the running time is $\Os{n+\frac{2n}{k-2}-1 \choose n}$,
which is $\Os(2^{\O((n/k)\log k)})$, as claimed.
Hence, it only remains to prove that the highest rank is at most $2n/(k-2)$.

For this, observe that in an $\ell$-good configuration, for each $0\le i< \ell$,
the union of all vertices having level $i$ and the set of $\khalf$ terminals together forms a cutset.
Since the graph is $k$-connected, this means that for each $0\le i<\ell$,
the number of vertices having level $i$ is at least $k/2-1$. The required bound on the rank easily follows.
\end{pfof}
\makeatletter{}%
\section{Upper Bounds for Spanning Tree Congestion}\label{sect:STC-upper}
We first state the following easy lemma, which together with Proposition~\ref{pr:min-max}, implies Lemma~\ref{lm:high-connected-STC-short}.
\begin{lemma}~\label{lm:high-connected-STC-full}
In a graph $G=(V,E)$, let $t_1$ be a vertex, and let $t_2,\cdots,t_\ell$ be any $(\ell-1)$ neighbours of $t_1$.
Suppose that there exists a $\ell$-connected-partition $\cup_{j=1}^\ell V_\ell$ such that for all $j\in \ell$, $t_j\in V_j$,
and the sum of degree of vertices in each $V_j$ is at most $D$.
Let $\tau_j$ be an arbitrary spanning tree of $G[V_j]$. Let $e_j$ denote the edge $\{t_1,t_j\}$.
Let $\tau$ be the spanning tree of $G$ defined as $\tau:=
\left(\cup_{j=1}^\ell~\tau_j\right)~\bigcup~ \left( \cup_{j=2}^\ell~e_j \right)$.
Then $\tau$ has congestion at most $D$.
\end{lemma}

\begin{theorem}\label{thm:congestion-faster}
For any connected graph $G=(V,E)$, there is an algorithm which computes a spanning tree with congestion at most $8\sqrt{mn}$ in 
$\Os\left(2^{\O\left(n\log n / \sqrt{m/n}\right)}\right)$ time.
\end{theorem}

\begin{theorem}\label{thm:congestion-polytime}
For any connected graph $G=(V,E)$, there is a polynomial time algorithm which computes a spanning tree with congestion at most $16\sqrt{mn} \cdot \log n$.
\end{theorem}

\newcommand{\bm}{\hat{m}}
\newcommand{\bn}{\hat{n}}
\newcommand{\nH}{n_{_H}}
\newcommand{\mH}{m_{_H}}

The two algorithms follow the same framework, depicted in Algorithm~\ref{alg:congestion}.
It is a recursive algorithm; the parameter $\bm$ is a global parameter,
which is the number of edges in the input graph $G$ in the first level of the recursion;
let $\bn$ denote the number of vertices in this graph.

The only difference between the two algorithms is in Line~\ref{ln:partition} on how this step is executed,
with trade-off between the running time of the step $T(\bm,\nH,\mH)$, and the guarantee $D(\bm,\nH,\mH)$.
For proving Theorem~\ref{thm:congestion-faster}, we use Theorem~\ref{thm:weighted-Gyori-Lovasz}(b),
Proposition~\ref{pr:min-max} and Lemma~\ref{lm:high-connected-STC-full},
yielding $D(\bm,\nH,\mH) \le 8\mH\sqrt{\nH/\bm}$ and $T(\bm,\nH,\mH) = \Os\left(2^{\O\left(\nH\log \nH / \sqrt{\bm/\nH}\right)}\right)$.
For proving Theorem~\ref{thm:congestion-polytime}, we make use of an algorithm in Chen et al.~\cite{CKLRSV2007},
which yields $D(\bm,\nH,\mH) \le 16\mH\sqrt{\nH/\bm}\cdot \log \nH$ and $T(\bm,\nH,\mH) = \poly(\nH,\mH)$.

\begin{algorithm}[H]
	\caption{\flcs($H,\bm$)}\label{alg:congestion}
	\SetKwInOut{Input}{Input}\SetKwInOut{Output}{Output}
	\Input{A connected graph $H=(V_H,E_H)$ on $\nH$ vertices and $\mH$ edges}
	\Output{A spanning tree $\tau$ of $H$}%
	\BlankLine
	\If {$\mH \le 8\sqrt{\bm \nH}$}{
		\Return an arbitrary spanning tree of $H$
	}
	$k\leftarrow \ceil{\sqrt{\bm/\nH}}$\\
	$Y\leftarrow$ a global minimum vertex cut of $H$\label{ln:min-vertex-cut}\\
	\eIf {$|Y|<k$}{
		$X\leftarrow$ the smallest connected component in $H[V_H\setminus Y]$
		(See Figure~\ref{fig:alg-low-conn})\\
		$Z\leftarrow~V_H\setminus (X\cup Y)$\\
		$\tau_1 \leftarrow$ \flcs( $H[X],\bm$ )\\
		$\tau_2 \leftarrow$ \flcs( $H[Y\cup Z],\bm$);     \textsf{    ($H[Y\cup Z]$ is connected as $Y$ is a global min cut}) \label{ln:connected}\\
		\Return $\tau_1 ~\cup~ \tau_2 ~\cup~$(an arbitrary edge between $X$ and $Y$)
	}
	{
		$t_1\leftarrow$ an arbitrary vertex in $V_H$\label{ln:begin}\\
		Pick $\floor{k/2}$ neighbours of $t_1$ in the graph $H$; denote them by $t_2,t_3,\cdots,t_{\floor{k/2}+1}$.
		Let $e_j$ denote edge $t_1t_j$ for $2\le j\le \floor{k/2}+1$.
		(See Figure~\ref{fig:alg-high-conn})\\
		Compute a $(\floor{k/2}+1)$-connected-partition of $H$, denoted by $\cup_{j=1}^{\floor{k/2}+1} V_j$, such that
		for each $j\in [\floor{k/2}+1]$, $t_j\in V_j$, and the total degree (w.r.t.~graph $H$) of vertices in each $V_j$ is at most $D(\bm,\nH,\mH)$.
		Let the time needed be $T(\bm,\nH,\mH)$.\label{ln:partition}\\
		For each $j\in [\floor{k/2}+1]$, $\tau_j\leftarrow $ an arbitrary spanning tree of $G[V_j]$\\
		\Return $\left(\cup_{j=1}^{\floor{k/2}+1}~\tau_j\right)~\bigcup~ \left( \cup_{j=2}^{\floor{k/2}+1}~e_j \right)$\label{ln:end}
	}
\end{algorithm}

\begin{figure}[h]
\centering
\includegraphics[scale=0.3]{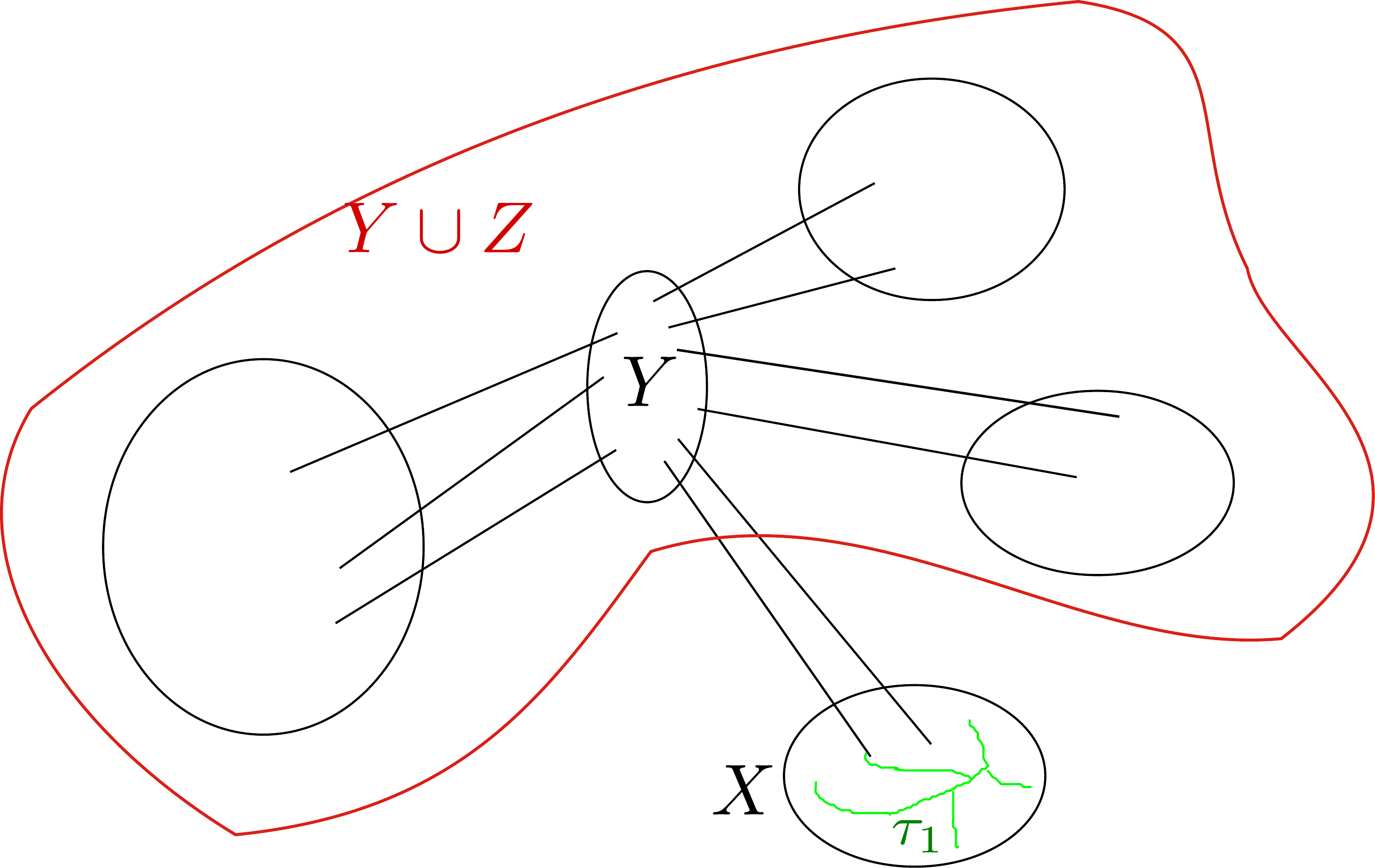}
\caption{The scenario in Algorithm~\ref{alg:congestion} when the graph has low connectivity.
The vertex set $Y$ is a global minimum vertex cut of the graph.
The vertex set $X$ is the smallest connected component after the removal of $Y$, and $Z$ is the union of all the other connected components.}\label{fig:UB}
\label{fig:alg-low-conn}
\end{figure}
\begin{figure}[h]
\centering
\includegraphics[scale=0.33]{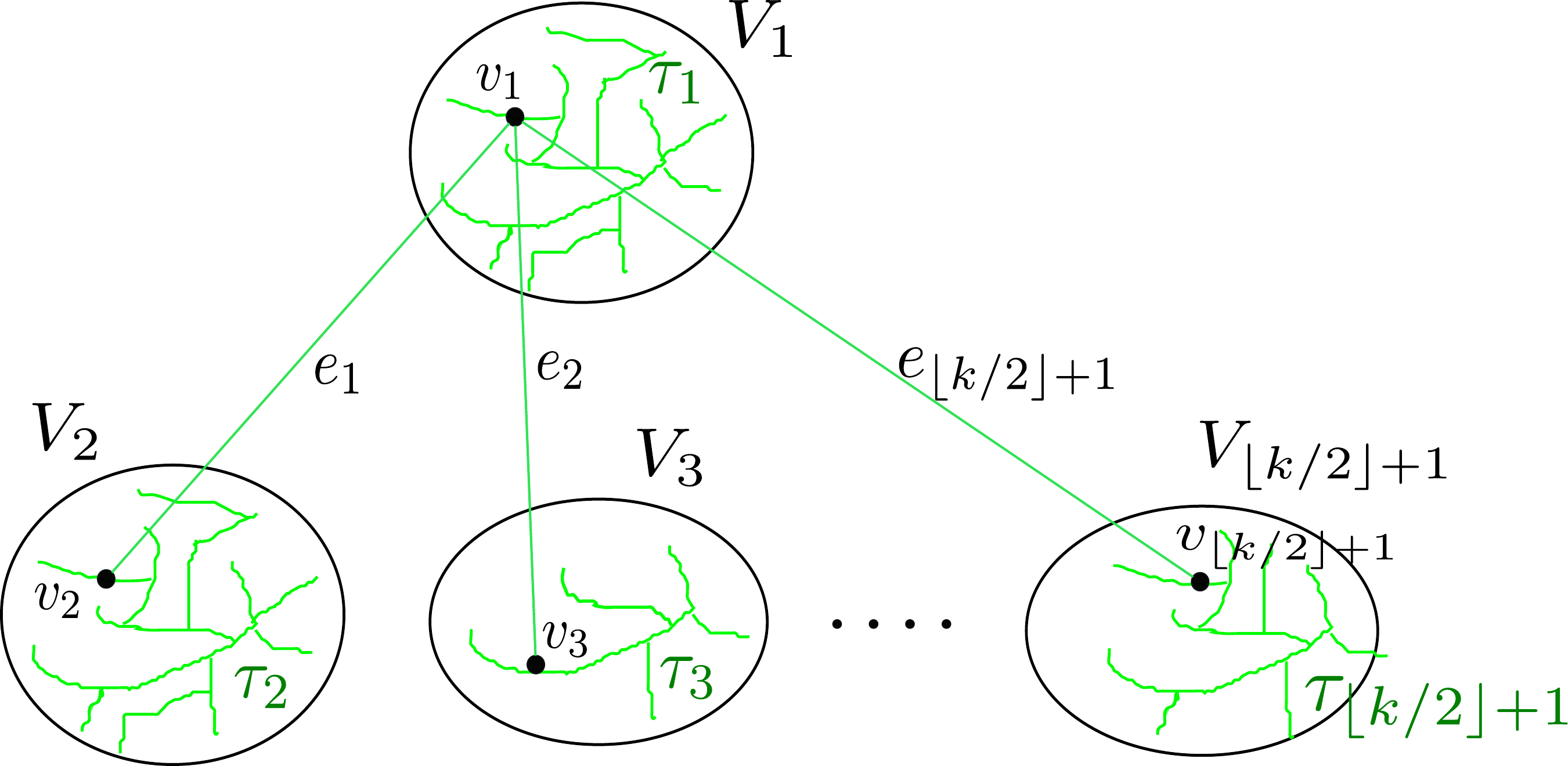}
\caption{The scenario in Algorithm~\ref{alg:congestion} when the graph has high connectivity.}
\label{fig:alg-high-conn}
\end{figure}

In the rest of this section, we first discuss the algorithm in Chen et al.,
then we prove Theorem~\ref{thm:congestion-polytime}.
The proof of Theorem~\ref{thm:congestion-faster} is almost identical, and is deferred to Appendix~\ref{app:sub-exp}.

\parabold{Single-Commodity Confluent Flow and The Algorithm of Chen et al.}
In a \emph{single-commodity confluent flow} problem, the input includes a graph $G=(V,E)$,
a \emph{demand} function $w:V\ra \rrplus$ and $\ell$ sinks $t_1,\cdots,t_\ell \in V$.
For each $v\in V$, a flow of amount $w(v)$ is routed from $v$ to one of the sinks.
But there is a restriction: at every vertex $u\in V$, the outgoing flow must leave $u$ on at most $1$ edge,
i.e., the outgoing flow from $u$ is unsplittable.
The problem is to seek a flow satisfying the demands which minimizes the \emph{node congestion}, i.e., the maximum incoming flow among all vertices.
Since the incoming flow is maximum at one of the sinks, it is equivalent to minimize the maximum flow received among all sinks.
(Here, we assume that no flow entering a sink will leave.)

\emph{Single-commodity splittable flow} problem is almost identical to single-commodity confluent flow problem,
except that the above restriction is dropped, i.e., now the outgoing flow at $u$ can split along multiple edges.
Note that here, the maximum incoming flow might not be at a sink.
It is known that single-commodity splittable flow can be solved in polynomial time.
For brevity, we drop the phrase ``single-commodity'' from now on.

\begin{theorem}[{\cite[Section 4]{CKLRSV2007}}]\label{thm:flow}
Suppose that given graph $G$, demand $w$ and $\ell$ sinks, there is a splittable flow with node congestion $q$.
Then there exists a polynomial time algorithm which computes a confluent flow with node congestion at most $(1+\ln \ell)q$ for the same input.
\end{theorem}

\begin{corollary}\label{co:flow}
Let $G$ be a $k$-connected graph with $m$ edges. Then for any $\ell\le k$ and for any $\ell$ vertices $t_1,\cdots,t_\ell\in V$,
there exists a polynomial time algorithm which computes an $\ell$-connected-partition $\cup_{j=1}^\ell V_\ell$
such that for all $j\in \ell$, $t_j\in V_j$, and the total degrees of vertices in each $V_j$ is at most $4(1+\ln \ell)m/\ell$.
\end{corollary}

Corollary~\ref{co:flow} follows from Theorem~\ref{thm:flow} and Proposition~\ref{pr:min-max}. See Appendix~\ref{app:flow} for details.

\parabold{Congestion Analysis.} We view the whole recursion process as a recursion tree.
There is no endless loop, since down every path in the recursion tree,
the number of vertices in the input graphs are \emph{strictly} decreasing.
On the other hand, note that the leaf of the recursion tree is resulted by either
(i) when the input graph $H$ to that call satisfies $\mH \le 8\sqrt{\bm \nH}$,
or (ii) when Lines~\ref{ln:begin}--\ref{ln:end} are executed.
An internal node appears only when the vertex-connectivity of the input graph $H$ is \emph{low},
and it makes two recursion calls.%

We prove the following statement by induction from bottom-up:
for each graph which is the input to some call in the recursion tree,
the returned spanning tree of that call has congestion at most $16\sqrt{\bm \nH}\log \nH$.

We first handle the two basis cases (i) and (ii). In case (i), \flcs~returns an arbitrary spanning tree,
and the congestion is bounded by $\mH\le 8\sqrt{\bm \nH}$.
In case (ii), by Corollary~\ref{co:flow} and Lemma~\ref{lm:high-connected-STC-full},
\flcs~returns a tree with congestion at most $16\mH\sqrt{\nH/\bm}\cdot \log \nH \le 16\sqrt{\bm\nH} \cdot \log \nH$.

Next, let $H$ be the input graph to a call which is represented by an internal node of the recursion tree.
Recall the definitions of $X,Y,Z,\tau_1,\tau_2$ in the algorithm.

Let $|X| = x$. Note that $1\le x\le \nH/2$. Then by induction hypothesis, the congestion of the returned spanning tree is at most
\begin{align}
&\max \{~\text{congestion of }\tau_1\text{ in }H[X]~,~\text{congestion of }\tau_2\text{ in }H[Y\cup Z]~\} ~+~ |X|\cdot |Y|\nonumber\\
~\le~& 16 \sqrt{\bm (\nH - x)} \log (\nH-x) ~+~ \left(\sqrt{\bm/\nH}+1\right) \cdot x.\label{eq:change-one}
\end{align}

Viewing $x$ as a real variable, by taking derivative, it is easy to see that the above expression is maximized at $x=1$.
Thus the congestion is at most
$$16\sqrt{\bm(\nH-1)}\log (\nH-1) + \sqrt{\bm/\nH}+1 ~~\le~~ 16\sqrt{\bm\nH}\log \nH,~~\text{as desired by Theorem~\ref{thm:congestion-polytime}.}$$

\parabold{Runtime Analysis.} At every internal node of the recursion tree, the algorithm makes two recursive calls with two vertex-disjoint and strictly smaller (w.r.t.~vertex size) inputs.
The dominating knitting cost is in Line~\ref{ln:min-vertex-cut} for computing a global minimum vertex cut,
which is well-known that it can be done in polynomial time.
Since at every leaf of the recursion tree the running time is polynomial, by standard analysis on divide-and-conquer algorithms,
the running time of the whole algorithm is polynomial, which completes the proof of Theorem~\ref{thm:congestion-polytime}.
\makeatletter{}%
\section{Lower Bound for Spanning Tree Congestion}\label{sect:STC-lower}

Here, we give a lower bound on spanning tree congestion which matches our upper bound.

\begin{theorem}\label{thm:LB-matching}
For any sufficiently large $n$, and for any $m$ satisfying $n^2/2 \ge m\ge \max\{16n\log n,100n\}$,
there exists a connected graph with $N=(3-o(1))n$ vertices and $M\in [m,7m]$ edges,
for which the spanning tree congestion is at least $\Omega\left( \sqrt{mn} \right)$.
\end{theorem}

We start with the following lemma, which states that for a random graph $\calG(n,p)$,
when $p$ is sufficiently large, its \emph{edge expansion} is $\Theta(np)$ with high probability.
The proof of the lemma uses only fairly standard arguments
and is deferred to Appendix~\ref{app:exist-expander}.

\begin{lemma}\label{lem:Hnm}
For any integer $n\ge 4$ and $1\ge p \ge 32 \cdot \frac{\log n}{n}$,
let $\calG(n,p)$ denote the random graph with $n$ vertices, in which each edge occurs independently with probability $p$.
Then with probability at least $1-\O(1/n)$, (i) the random graph is connected,
(ii) the number of edges in the random graph is between $pn^2/4$ and $pn^2$,
and (iii) for each subset of vertices $S$ with $|S|\le n/2$, the number of edges leaving $S$ is
at least $\frac p2 \cdot |S| \cdot (n-|S|)$.
\end{lemma}

In particular, for any sufficiently large integer $n$, when $n^2/2\ge m\ge 16n\log n$, by setting $p=2m/n^2$,
there exists a connected graph with $n$ vertices and $[m/2,2m]$ edges, such that
for each subset of vertices $S$ with $|S|\le n/2$, the number of edges leaving $S$ is at least $\frac{m}{2n}\cdot |S| = \Theta(m/n)\cdot |S|$.
We denote such a graph by $H(n,m)$.

\begin{figure}[h]
	\centering
	\includegraphics[scale=0.45]{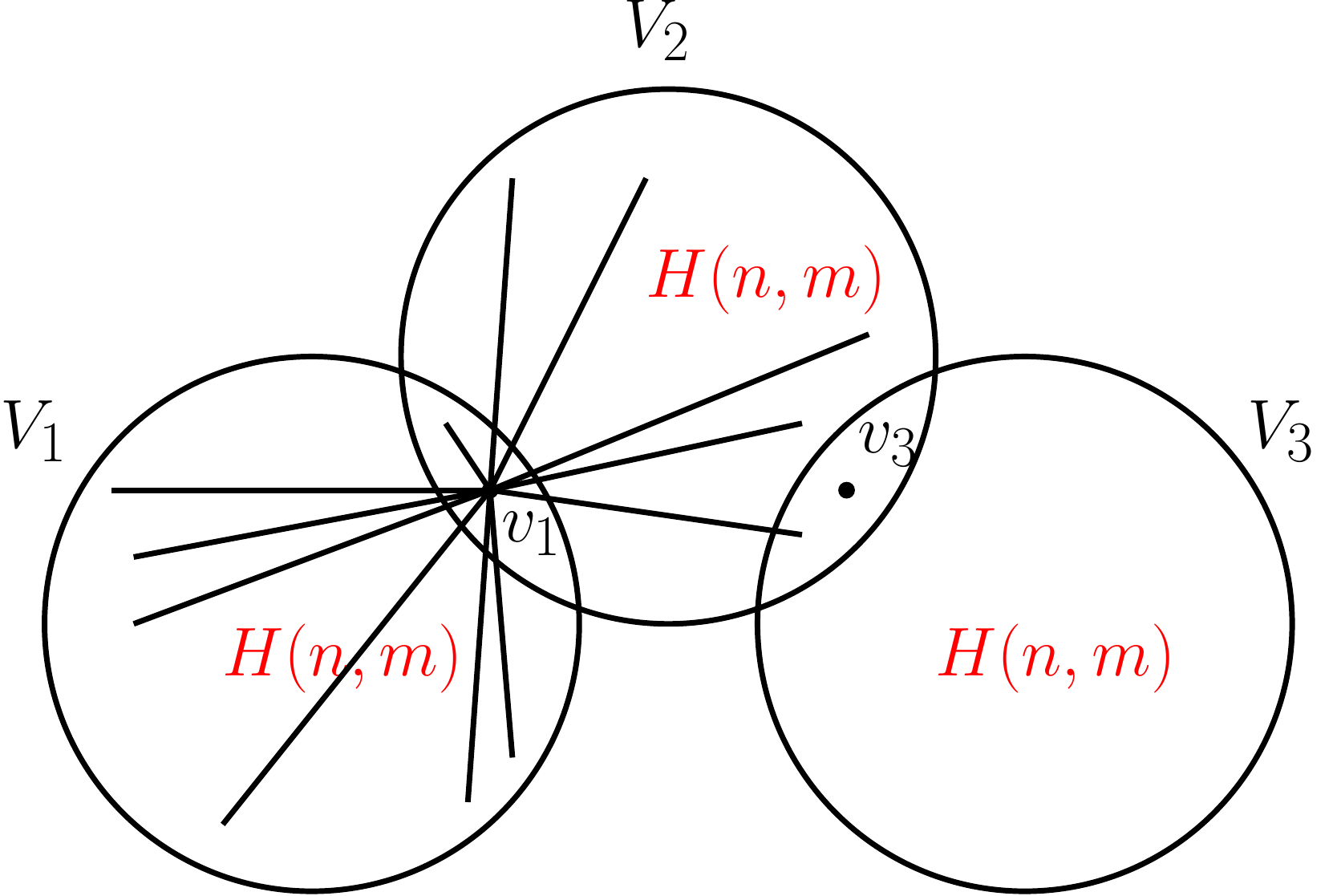}
	\caption{Our lower-bound construction for spanning tree congestion. $V_1,V_2,V_3$ are three vertex subsets of the same size. In each of the subsets, we embed expander $H(n,m)$.
	There is a small overlap between $V_2$ and $V_1,V_3$, while $V_1,V_3$ are disjoint.
	For any vertex $v_1\in V_1\cap V_2$, we add edges between it and any other vertex in $V_1\cup V_2$;
	similarly, for any vertex $v_3\in V_3\cap V_2$, we add edges (not shown in figure) between it and any other vertex in $V_3\cup V_2$.}
	\label{fig:LB}
\end{figure}

We discuss our construction here (see Figure~\ref{fig:LB}) before delving into the proof.
The vertex set $V$ is the union of three vertex subsets $V_1,V_2,V_3$,
such that $|V_1| = |V_2| = |V_3| = n$, $|V_1\cap V_2| = |V_2\cap V_3| = \sqrt{m/n}$, and $V_1,V_3$ are disjoint.
In each of $V_1$, $V_2$ and $V_3$, we embed $H(n,m)$.
The edge sets are denoted $E_1,E_2,E_3$ respectively.
Up to this point, the construction is similar to that of Ostrovskii~\cite{Ostrovskii2004}, except that we use $H(n,m)$ instead of a complete graph.

The new component in our construction is adding the following edges.
For each vertex $v\in V_1\cap V_2$, add an edge between $v$ and every vertex in $(V_1\cup V_2)\setminus \{v\}$.
The set of these edges are denoted $F_1$.
Similarly, for each vertex $v\in V_3\cap V_2$, add an edge between $v$ and every vertex in $(V_3\cup V_2)\setminus \{v\}$.
The set of these edges are denoted $F_3$.
This new component is crucial: without it, we could only prove a lower bound of $\Omega(m/\sqrt{n}) = \Omega(\sqrt{mn}\cdot\frac{\sqrt{m}}{n})$.

\begin{pfof}{Theorem~\ref{thm:LB-matching}}
Let $G=(V,E)$ be the graph constructed as above.
The whole graph has $3n - 2\sqrt{m/n}$ vertices. The number of edges is at least $m$ (due to edges in $E_1$ and $E_3$),
and is at most $6m+2\sqrt{m/n} \cdot 2n = 6m + 4\sqrt{mn}$, which is at most $7m$ for all sufficiently large $n$.

It is well known that for any tree on $n$ vertices, there exists a vertex $x$ called a centroid of the tree such that, removing $x$ decomposes the tree into connected components, each of size at most $n/2$.
Now, consider any spanning tree of the given graph, let $u$ be a centroid of the tree.
Without loss of generality, we can assume that $u\notin V_1$; otherwise we swap the roles of $V_1$ and $V_3$.
The removal of $u$ (and its adjacent edges) from the tree decomposes the tree into a number of connected components.
For any of these components which intersects $V_1$, it must contain at least one vertex of $V_1\cap V_2$,
thus the number of such components is at most $\sqrt{m/n}$, and hence there exists one of them, denoted by $U_j$,
such that
$$b_1 ~:=~ |U_j\cap V_1| ~\ge~ n/(\sqrt{m/n}) ~=~ n\sqrt{n/m}.$$
Let $e_j$ denote the tree-edge that connects $u$ to $U_j$.
Then there are three cases:

\parabold{Case 1: $n\sqrt{n/m} ~\le~ b_1 ~\le~ n-n\sqrt{n/m}$.}~~
Due to the property of $H(n,m)$, the congestion of $e_j$ is at least $\Theta(m/n)\cdot \min\{b_1,n-b_1\} ~\ge~ \Theta(\sqrt{mn})$.

\medskip

\parabold{Case 2: $b_1 ~>~ n-n\sqrt{n/m}$ and $|U_j \cap V_1\cap V_2|~\le~ \frac 12 \cdot \sqrt{m/n}$.}~~
Let $W := (V_1\cap V_2)\setminus U_j$. Note that by this case's assumption, $|W_1|~\ge~\frac 12 \cdot \sqrt{m/n}$.
Due to the edge subset $F_1$, the congestion of $e_j$ is at least
$$\Big|F_1(W~,~V_1\setminus W)\Big|~\ge~ \left(\frac 12 \cdot \sqrt{m/n}\right) \cdot \frac{n}{2} ~=~ \Theta\left(\sqrt{mn}\right).$$

\medskip

\parabold{Case 3: $b_1 ~>~ n-n\sqrt{n/m}$ and $|U_j \cap V_1\cap V_2|~>~ \frac 12 \cdot \sqrt{m/n}$.}~~
Let $W' := U_j \cap V_1\cap V_2$, and let $Z := (V_2\setminus V_1)\cap U_j$.

Note that $b_1 ~>~ n-n\sqrt{n/m} ~\ge~ 9n/10$. Suppose $|Z| ~\ge~ 6n/10$, then $|U_j|~>~ 9n/10 + 6n/10 > |V|/2$,
a contradiction to the assumption that $u$ is a centroid. Thus, $|Z| ~<~ 6n/10$.
Due to the edge subset $F_2$, the congestion of $e_j$ is at least
\begin{align*}
\Big|F_2(W'\cup Z~,~V_2\setminus (W'\cup Z))\Big| &~\ge~ |W'| \cdot \left(n-|W'|-|Z|\right)\\
&~\ge~ \left(\frac 12 \cdot \sqrt{m/n} \right) \cdot \left( n - \sqrt{m/n} - \frac{6n}{10} \right) ~=~ \Theta(\sqrt{mn}).\qedhere
\end{align*}
\end{pfof}

\makeatletter{}%
\section{Graphs with Expanding Properties}\label{sec:expanding}
\newcommand{\bvt}{\overline{V_T}}
\newcommand{\vt}{V_T}
\newcommand{\bs}{\overline{S}}
\newcommand{\ba}{\overline{A}}

For any vertex subset $U,W\subset V$, let $N_W(U)$ denote the set of vertices in $W$ which are adjacent to a vertex in $U$.
Let $N(U) := N_{V\setminus U}(U)$.

\begin{definition}\label{def:expanding}
A graph $G=(V,E)$ on $n$ vertices is an $(n,s,d_1,d_2,d_3,t)$-expanding graph if the following four conditions are satisfied:
\begin{enumerate}
\item[(1)] for each vertex subset $S$ with $|S| = s$, $|N(S)|\ge d_1 n$;
\item[(2)] for each vertex subset $S$ with $|S|\le s$, $|N(S)|\ge d_2 |S|$;
\item[(3)] for each vertex subset $S$ with $|S|\le n/2$ %
and for any subset $S'\subset S$, %
$|N_{V\setminus S}(S')| ~\ge~ |S'|-t$.
\item[(4)]For each vertex subset $S$, $\left|E(S,V\setminus S)\right|\le d_3|S|$.
\end{enumerate}
\end{definition}

\begin{theorem}\label{thm:expanding-UB}
For any connected graph $G$ which is an $(n,s,d_1,d_2,d_3,t)$-expanding graph, there is a polynomial time algorithm which computes a spanning tree
with congestion at most
$$
d_3 \cdot \left[4 \cdot \max\left\{s+1~,~\ceil{\frac{3d_1n}{d_2}}\right\} \cdot \left( \frac{1}{2d_1} \right)^{\log_{(2-\delta)}2} + t\right],~~\text{where}~
\delta = \frac{t}{d_1 n}.
$$
\end{theorem}

Next, we present the polynomial time algorithm in Theorem~\ref{thm:expanding-UB} and its analysis.

\parabold{Algorithm.} Let $G$ be an $(n,s,d_1,d_2,d_3,t)$-expanding graph.
By Condition (2), every vertex has degree at least $d_2$.
Let $v_0$ be a vertex of degree $d\ge d_2$, and let $v_1,\cdots,v_d$ be its $d$ neighbours.
We maintain a tree $T$ rooted at $v_0$ such that 
$T=T_1\cup T_2 \cup \dots \cup T_d \cup \left\{ v_0v_1,v_0v_2,\dots,v_0v_d \right\}$ 
where $T_1,T_2,\cdots,T_d$ are trees rooted at $v_1,v_2,\dots,v_d$ respectively.
We call the $T_i's$ as \underline{branches}. (See Figure~\ref{fig:random-tree}).
We start with each branch $T_i=v_i$.
In order to minimize congestion, we grow $T$ in a balanced way, i.e., we maintain that the $T_i$'s are roughly of the same size.
A branch is \underline{saturated} if it contains at least $\max \left\{s+1~,~\frac{ 3d_1n}{d_2}\right\}$ vertices.%

\begin{figure}[h]
\centering
\includegraphics[scale=0.4]{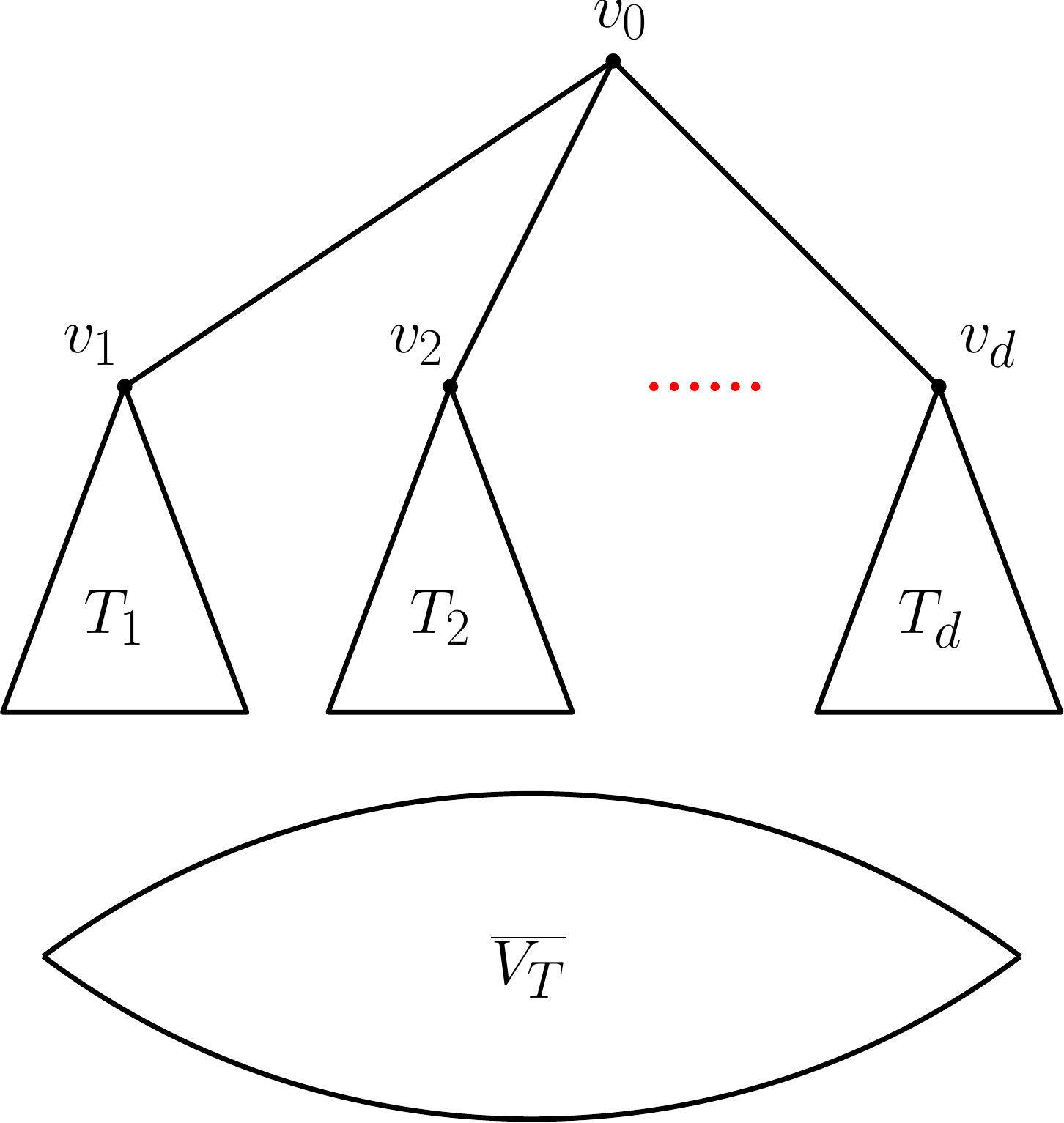}
\caption{The tree $T$ and its branches} 
\label{fig:random-tree}
\end{figure}

At any point of time, let $\vt$ be the set of vertices in $T$ and $\bvt$ be the vertices not in $T$.
Often, we will move a subtree of a saturated branch $T_i$ to an unsaturated branch $T_j$ to ensure balance.
For any $x\in \vt$, let $T_{x}$ denote the subtree of $T$ rooted at $x$.
A vertex $x$ of a saturated branch $T_i$ is called \underline{transferable} (to branch $T_j$)
if $x$ has a neighbour $y$ in $T_j$ and the tree $T_j\cup \{xy\} \cup T_x$ is unsaturated. (See Figure~\ref{fig:transfer}.)

\begin{figure}[h]
\centering
\includegraphics[scale=0.4]{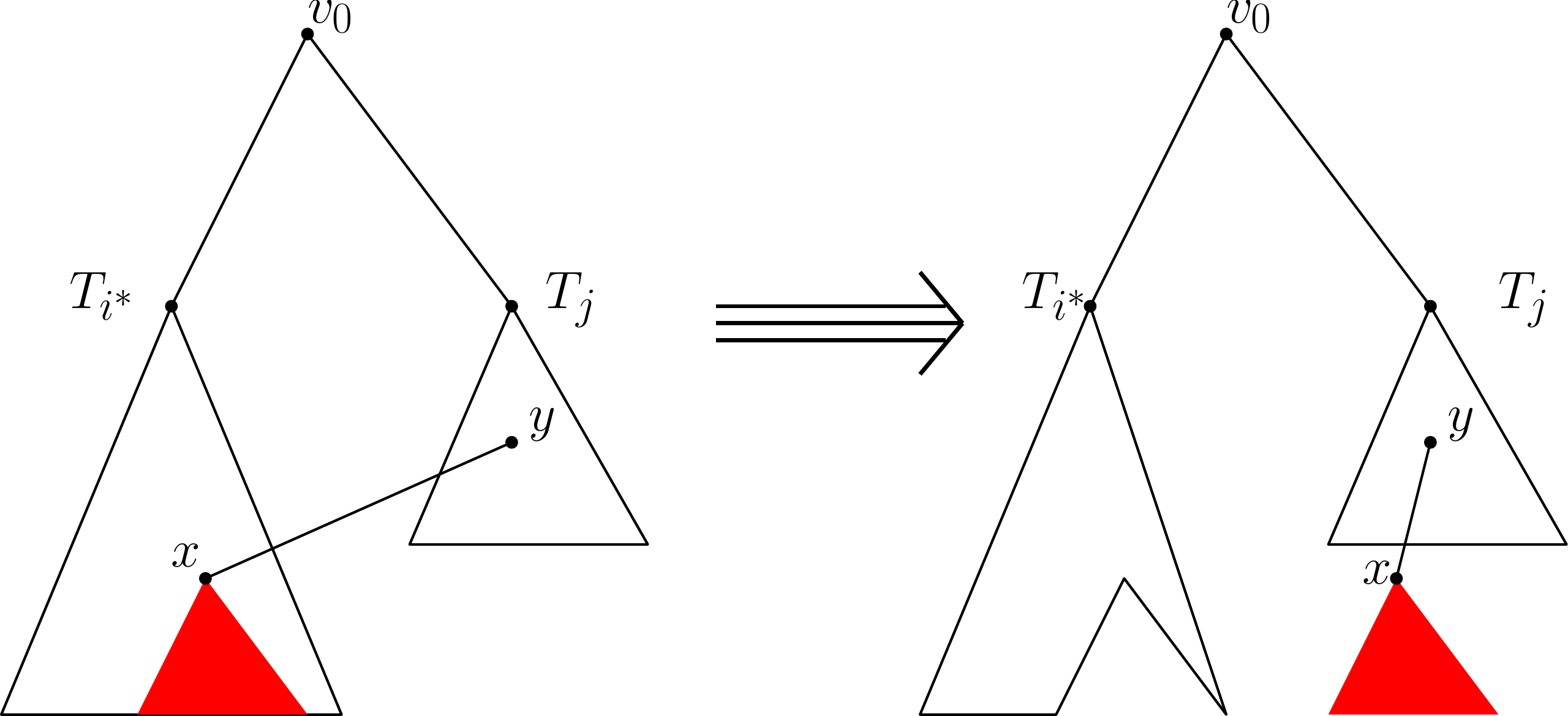}
\caption{Transfer of a subtree from a saturated branch to an unsaturated branch}\label{fig:transfer}
\end{figure}

The algorithm is divided into two phases which are described below.
Throughout the algorithm, whenever a branch $T_i$ gets modified, $T$ gets modified accordingly,
and whenever $T$ gets modified $\vt$ and $\bvt$ gets modified accordingly.

\smallskip

\parabold{Phase 1}: Repeatedly do one of the following two actions, until $|\vt|\ge d_1n$:\\
(We will prove %
that the precondition of at least one of the actions is satisfied if $|\vt|< d_1n$)%
\begin{enumerate}
	\item If there exists a $b\in \bvt$ such that $b$ has a neighbour $a$ in some unsaturated branch $T_i$:\\
	Add the vertex $b$ and the edge $ab$ to branch $T_i$.
	\item If there exists at least one transferable vertex: (see Figure~\ref{fig:transfer})\\
	Find the transferable vertex $x$ such that $T_x$ is the smallest. Let $T_{i^*}$ be the branch currently containing $x$, $T_j$ be a branch to which it is transferable, and $y$ be an arbitrarily chosen neighbour of $x$ in $T_j$.
	\begin{enumerate}
		\item
		Remove the subtree $T_x$ from $T_{i^*}$ and add it to $T_j$ with $x$ as a child of $y$.
		\item
		Pick a $b\in \bvt$ that has a neighbour $a$ (arbitrarily chosen, if many) either in $T_{i^*}$ or in $T_j$. (We will show in the analysis that such $b$ exists). 
		We add vertex $b$ and edge $ab$ to the branch containing $a$ (i.e. to $T_{i^*}$ or $T_j$).
	\end{enumerate}
\end{enumerate}
\parabold{Phase 2:}
While $\bvt\neq \emptyset$, repeat:	\\
Find a maximum matching of $G[\vt,\bvt]$, the bipartite graph formed by edges of $G$ between $\vt$ and $\bvt$. Let $M$ be the matching.
Add all edges of $M$ to $T$.

\smallskip
In the analysis below, we say that a tree is \underline{saturated} if it contains at least $\calA$ vertices;
we will determine its appropriate value by the end of the analysis.

\parabold{Analysis of Phase 1.} 
We claim that during Phase 1, i.e. if $|\vt|< d_1 n$, the precondition of either step 1 or step 2 is satisfied. We also show the existence of a vertex $b$ as specified in step 2b, whenever step 2b is reached. 
Given these and the fact that a vertex in $\bvt$ is moved to $\vt$ (either in step 1 or in step 2b) during each round of Phase 1, we have that Phase 1 runs correctly and terminates after a linear number of rounds. 

During Phase 1, we will also maintain the invariant that each branch has at most $\calA$ vertices; thus, each saturated branch
has exactly $\calA$ vertices. We call this invariant the \underline{balancedness}.
Note that balancedness is not violated due to step 1, as the new vertex is added to an unsaturated branch. 
It is not violated during step 2 as the branches $T_{i^*}$ and $T_j$ (as defined in step 2) become unsaturated at the end of the step.

We define the \underline{hidden vertices} of $T$ (denoted by $H\equiv H_T$) as follows:
they are the vertices which are not adjacent to any vertices outside the tree, i.e., to any vertex in $\bvt$.
If there is an unsaturated branch with a non-hidden vertex, clearly the precondition of step 1 is satisfied.
So, let us assume that all the vertices in all unsaturated branches are hidden.
In such a case, we show that the precondition of step 2 is satisfied if $|\vt|< d_1n$.

We argue that in this case $|H|\le s$: otherwise, take a subset $H'\subset H$ of cardinality $s$,
then by condition (2), $N(H')$, which is contained in $V_T$, has cardinality at least $d_1 n$, a contradiction.

Since $|\vt|<d_1 n$, the number of saturated branches is at most $d_1n/\calA$.
To ensure that at least one unsaturated branch exists, we set $\calA$ such that $d_1n/\calA < d_2$.
Let $U$ denote the set of vertices in all unsaturated branches. Since all vertices in $U$ are hidden vertices, $|U|\le s$.
Then by condition (2), $|N(U)|\ge d_2|U|$. Note that the vertices in $N(U)$ are all in the saturated branches.
By the pigeon-hole principle, there exists a saturated branch containing at least
$$N(U) / (d_1n/\calA) ~\ge~ \frac{\calA d_2|U|}{d_1n}$$
vertices of $N(U)$. 
By setting $\calA \ge \frac{ 3d_1n}{d_2} $, the above calculation guarantees the existence of a saturated branch containing
at least $3|U|\ge |U|+2$ vertices of $N(U)$; let $T_i$ be such a branch.

In $T_i$, pick a vertex $x\in T_i\cap N(U)$ such that $T_x$ does not contain any vertex in $N(U)$, except $x$.
Then the size of $T_x$ is at most $\calA-|N(U)\cap T_i|+1\ge \calA - (|U|+1)$.
Let $y\in U$ be a vertex which is adjacent to $x$ and $T_j$ be the branch containing $y$.
Since $T_j$ has at most $|U|$ vertices, $x$ is a transferable vertex (to $T_j$).
Thus precondition of step 2 is satisfied.

We further set $\calA > s$ so that in each saturated branch, there is at least one unhidden vertex. %
In particular, $T_i$ has an unhidden vertex, which is adjacent to some $b\in \bvt$.
The vertex $b$ is either adjacent to a vertex in $T_x$, or a vertex in $T_i \setminus T_x$ as required in step 2b.

\smallskip

\parabold{Analysis of Phase 2.} 
Since $G$ is connected, $M$ is non-empty in each iteration of Phase 2, and hence Phase $2$ terminates in linear number of rounds. 
At the end of Phase $2$, since $\bvt$ is empty, $T$ is clearly a spanning tree.
It only remains to estimate the congestion of this spanning tree.
Towards this, we state the following \emph{modified Hall's theorem}, which is an easy corollary of the standard Hall's theorem.

\begin{lemma}
	In a bipartite graph $(L,R)$ with $|L|\le |R|$, for any vertex $w\in L$, let $R(w)$ denote the neighbours of $w$ in $R$;
	then for any $W\subset L$, let $R(W) := \cup_{w\in W} R(w)$.
	Suppose that there exist $t\ge 0$ such that for any $W\subset L$, we have $|R(W)|\ge |W|-t$.
	Then the bipartite graph admits a matching of size at least $|L| - t$.
\end{lemma}

Recall that Phase $2$ consists of multiple rounds of finding a matching between $V_T$ and $\bvt$.
As long as $|V_T|\le n/2$,
condition (3) (with $S = V_T$) plus the modified Hall's theorem (with $L=\vt$ and $R=\bvt$) guarantees that in each round, at least
$$|V_T|-t ~\ge~ \left(1-\frac{t}{d_1n}\right)\cdot |V_T| ~=:~ (1-\delta) |V_T|$$
number of vertices in $\bvt$ are matched. Thus, after at most $\ceil{\log_{(2-\delta)} \frac{1}{2d_1}}$ rounds of matching, $|V_T|\ge n/2$.
After reaching $|V_T|\ge n/2$, condition (3) (with $S=\bvt$) plus the modified Hall's theorem (with $L=\bvt$ and $R=\vt$) guarantees that after one more round of matching, all but $t$ vertices are left in $\bvt$.

By the end of Phase 1, each branch had at most $\calA$ vertices.
After each round of matching, the cardinality of each branch is doubled at most.
Thus, the maximum possible number of vertices in each branch after running the whole algorithm is at most
$$
\calA \cdot 2^{\ceil{\log_{(2-\delta)} \frac{1}{2d_1}}+1} + t %
~~\le~~ 4\calA \cdot \left( \frac{1}{2d_1} \right)^{\log_{(2-\delta)}2} + t.
$$
and hence the STC is at most
$$
d_3 \cdot \left[4\calA \cdot \left( \frac{1}{2d_1} \right)^{\log_{(2-\delta)}2} + t\right] .
$$

Recall that we need $\calA$ to satisfy $d_1n/\calA < d_2$, $\calA \ge \frac{ 3d_1n}{d_2} $ and $\calA > s$.
Thus we set $\calA := \max \left\{ s+1~,~\ceil{\frac{ 3d_1n}{d_2}}\right\}$.
\subsection{Random Graph}
Let $G \in \calG(n,p)$ where $p \ge c_0\log n/n$, and $c_0=64$.
The following lemmas show that 
with high probability $G$ is an $(n,s,d_1,d_2,d_3,t)$-expanding graph with
$s=\Theta(1/p)$, $d_1 = \Theta(1)$, $d_2 = \Theta(np)$, $d_3 = \Theta(np)$, $t = \Theta(1/p)$ (and hence $\delta = o(1)$).
The proof of the lemmas are deferred to Appendix~\ref{ssec:algo-random}.

\begin{lemma}\label{lemma:big-set-expansion}
	For any $S\subseteq V(G)$ such that $ |S|=\lceil 1/p \rceil$, we have $|N(S)|\ge c_2n$ with probability at least $1-e^{-n/16}$, where $c_2=1/25$.
\end{lemma}
\begin{lemma}
	\label{lemma:small-set-expansion}
	For any $S\subseteq V(G)$ such that $|S|\le 1/p$, we have $ |N(S)|\ge c_3np|S|$ with probability at least $1-\O(1/n^2)$, where $c_3=1/16$.
\end{lemma}

\begin{lemma}\label{lemma:halls}
For all $A\subseteq V(G)$ such that $|A|\le n/2$, and for all $S\subseteq A$,  with probability at least $1-e^{-n}$, $S$ has at least $|S|-c_4/p$
neighbors in $V\left( G \right)\setminus A$, where $c_4=12$.
\end{lemma}

\begin{lemma}\label{lemma:smallcutsize}
	For all $S\subseteq V(G)$, the cut size $|E(S,V(G)\setminus S)|$ is at most $np|S|$ with probability at least $1-n^{-c_0/4}$.
\end{lemma}
Plugging the bounds from above lemmas into Theorem~\ref{thm:expanding-UB}, together with a separate lower bound argument (Theorem~\ref{thm:LB-random} in Appendix~\ref{app:random-non-algo}),
we have the following theorem; in Appendix~\ref{app:random-non-algo}, we also present a non-algorithmic proof of this theorem.

\begin{theorem}\label{thm:UB-random-algo}
If $G \in \calG(n,p)$ where $p \ge 64\log n/n$, then with probability at least $1-\O(1/n)$, its STC is $\Theta(n)$.
\end{theorem}
\makeatletter{}%
\section{Discussion and Open Problems}

In this paper, we provide thorough understanding, both combinatorially and algorithmically,
on the spanning tree congestion of general graphs and random graphs. 
On course of doing so, we also provide the first constructive proof for the generalized Gy\H{o}ri-Lov\'{a}sz theorem, which might be of independent interest.
Following are some natural open problems:
\begin{itemize}%
\item Finding the spanning tree with minimum congestion is \NP-hard; indeed, Bodlaender et al.~\cite{BFGOL2012}
showed a $(9/8-\epsilon)$-approximation \NP-hardness for the $\STC$ problem.
Does a constant or a poly-logarithmic factor approximation polynomial time algorithm exist?
\item We present an algorithm for computing a spanning tree achieving congestion at most $\O(\sqrt{mn})$. 
The algorithm runs in sub-exponential time when $m=\omega(n\log^2n)$. Is there a polynomial time algorithm for constructing such a spanning tree?
\item For a $k$-connected graph, a connected $k$-partition where all parts are of size at most $\O((n/k)\log k)$
can be found in polynomial time due to an algorithm of Chen et al.~\cite{CKLRSV2007}.
Can we improve the sizes of parts to $\O(n/k)$?
\item Is finding Gy\H{o}ri-Lov\'{a}sz partition \textsf{PLS}-complete? If not, is it polynomial time solvable? 
\end{itemize}

\bibliography{bib}

\newpage

\appendix

\makeatletter{}%
\section{Missing Proofs in Sections~\ref{sect:STC-upper} and~\ref{sect:STC-lower}}\label{app:congestion-UB}

\subsection{Proof of Corollary~\ref{co:flow}}\label{app:flow}

First of all, we set the demand of each vertex in the flow problem to be the the degree of the vertex in $G$,
and $t_1,\cdots,t_\ell$ as the sinks in the flow problem.

By Proposition~\ref{pr:min-max}, there exists an $\ell$-connected-partition $\cup_{j=1}^\ell U_\ell$
such that for all $j\in [\ell]$, $t_j\in U_j$, and the total degrees of vertices in each $U_j$ is at most $4m/\ell$.
With this, by routing the demand of a vertex in $U_j$ to $t_j$ via an arbitrary path in $G[U_j]$ only,
we construct a splittable flow with node congestion at most $4m/\ell$.
By Theorem~\ref{thm:flow}, one can construct a confluent flow with node congestion at most $4(1+\ln \ell)m/\ell$ in polynomial time.

Obviously, in the confluent flow, all the flow originating from one vertex goes completely into one sink.
Set $V_j$ to be the set of vertices such that the flows originating from these vertices go into $t_j$.
It is then routine to check that $\cup_{j=1}^\ell V_\ell$ is our desired $\ell$-connected-partition.

\subsection{Proof of Theorem~\ref{thm:congestion-faster}}\label{app:sub-exp}

Instead of giving the full proof, we point out the differences from the proof of Theorem~\ref{thm:congestion-polytime}.

First, in handling the basis case (ii), by Theorem~\ref{thm:weighted-Gyori-Lovasz}(b),
Proposition~\ref{pr:min-max} and Lemma~\ref{lm:high-connected-STC-full},
we have an improved upper bound on the congestion of the returned tree, which is $8\mH / \sqrt{\bm/\nH} \le 8\sqrt{\bm\nH}$.
Thus, \eqref{eq:change-one} can be improved to
$$8 \sqrt{\bm (\nH - x)} ~+~ \sqrt{\frac{\bm}{\nH}} \cdot x.$$
Again, by viewing $x$ as a real variable and taking derivative, it is easy to see that the above expression is maximized at $x=1$.
So the above bound is at most
$$
8 \sqrt{\bm (\nH - 1)} ~+~ \sqrt{\frac{\bm}{\nH}}~~\le~~ 8\sqrt{\bm\nH},~~\text{as desired.}
$$

Concerning the running time, it is clear that in the worst case,
it is dominated by some calls to the algorithm in Theorem~\ref{thm:weighted-Gyori-Lovasz}(b).
Note that the number of such calls is at most $\bn$, since each call to the algorithm is on a disjoint set of vertices.

\medskip

There remains one concern, which is the connectedness of $H[Y\cup Z]$.
Suppose the contrary that $H[Y\cup Z]$ is not connected.
Let $C$ be one of its connected components, so that it contains the least number of vertices from $Y$.
Then $C$ contains at most $\floor{|Y|/2}$ vertices from $Y$, i.e., $|C\cap Y| < |Y|$.
Note that $C\cap Y$ is a vertex cut set of the graph $H$, thus contradicting that $Y$ is a global minimum vertex cut set.

\subsection{Proof of Lemma~\ref{lem:Hnm}}\label{app:exist-expander}
It is well known that the requirements (i) and (ii) are satisfied with probability $1-o(1/n)$.~\cite{Bollobas-book}
For each subset $S$ with $|S|\le n/2$, by the Chernoff bound,
$$
\prob{~\Big|E(S,V\setminus S)\Big| ~\le~ \frac p2 \cdot |S| \cdot (n-|S|)~} ~\le~ e^{-p|S|(n-|S|)/8} ~\le~ e^{-pn|S|/16}.
$$
Since $p \ge 32 \cdot \frac{\log n}{n}$, the above probability is at most $n^{-2|S|}$.
Then by a union bound, the probability that (iii) is not satisfied is at most
$$
\sum_{s=1}^{\floor{n/2}} \binom{n}{s} \cdot n^{-2s} ~\le~ \sum_{s=1}^{\floor{n/2}} n^s \cdot n^{-2s}
~\le~ \sum_{s=1}^{\floor{n/2}} n^{-s} ~\le~ \frac 2n.$$
\makeatletter{}%

\section{Spanning Tree Congestion of Random Graphs}\label{app:random}

\subsection{Non-Algorithmic Proof of Theorem~\ref{thm:UB-random-algo}}\label{app:random-non-algo}

We first present a simple non-algorithmic proof that random graph has STC $\Theta(n)$ with high probability.
Theorem~\ref{thm:UB-random} gives the upper bound and Theorem~\ref{thm:LB-random} gives the lower bound.
The proof of Theorem~\ref{thm:UB-random} uses Lemma~\ref{lm:high-connected-STC-short}
and the fact that for random graphs, vertex-connectivity and minimum degree are equal with high probability.
Theorem~\ref{thm:UB-random} does not give an efficient algorithm.

\begin{theorem}\label{thm:UB-random}
If  $G \in \calG(n,p)$ where $p \ge 8\log n/n$,
then the spanning tree congestion of $G$ is at most $16n$ with probability at least $1-o(1/n)$. 
\end{theorem}

\begin{proof}
It is known that the threshold probability for a random graph being $k$-connected is same as the threshold probability for it having minimum degree at least $k$~\cite{bollobas1985random}. Since $p\ge 8\log n/n$, using Chernoff bound and taking union bound over all vertices gives that $G$ has minimum degree at least $np/2$ with probability at least $1-o(1/n)$. 
Hence $G$ is $(np/2)$-connected with probability at least $1-o(1/n)$. 
We also have that the number of edges in $G$ is at most $2n^2p$ with probability at least $1-o(1/n)$.
Now, by using Lemma~\ref{lm:high-connected-STC-short}, we have that with probability at least $1-o(1/n)$, the spanning tree congestion is at most $16n$. 
\end{proof}

\begin{theorem}\label{thm:LB-random}
If $G \in \calG(n,p)$ where $p \ge 32 \log n /n$, then the spanning tree congestion of $G$ is $\Omega(n)$ with probability $1-\O(1/n)$.
\end{theorem}
\begin{proof}
	By using Chernoff Bounds and applying union bound, it is easy to show that
	with probability $1-o(1/n)$, every vertex of $G$ has degree at most $c_1np$ for a sufficiently large constant $c_1$.
	Also, by Lemma~\ref{lem:Hnm}, with probability $1-\O(1/n)$, properties (i) and (iii) of that lemma holds.
	In the proof below, we conditioned on the above mentioned highly probable events.

	Take a spanning tree $T$  of $G$ which gives the minimum congestion.   
	Let $u$ be a centroid of the tree $T$, i.e., each connected component of 
	$T\setminus \left\{ u \right\}$ has at most $n/2$ vertices.
	If there is a connected component with number of vertices at least $n/4$,
	then define this connected component as $T'$. 
	Else, all connected components have at most $n/4$ vertices.
	In this case, let $T'$ be the forest formed by the union of a minimum number of connected components
	of $T\setminus \{u\}$ such that $|T'| \ge n/4$. It is easy to see that $|T'| \le n/2$.
	Also, the number of edges in $T$ from $V(T')$ to $V(T)\setminus V(T')$
	is at most $\deg_G(u)$, which is at most $ c_1np$.
	
By property (iii) of Lemma~\ref{lem:Hnm},
the number of edges between $V(T')$ and $V(G) \setminus V(T')$ is $\Omega(n^2p)$.
Each of these edges in $G$ between $V(T')$ and $V(G)\setminus V(T')$ have
to contribute to the congestion of at least one of the edges in $T$ between $V(T')$ and $V(G)\setminus V(T')$.
Now since $T'$ sends at most $c_1np$ tree edges to other parts of $T$,  
it follows that there exists one edge in $T$ with congestion at least
$\Omega(n^2p) / (c_1np) = \Omega(n)$, as claimed.
\end{proof}

\subsection{Random Graph Satisfies Expanding Properties}\label{ssec:algo-random}

\noindent\textbf{Constants.} For easy reference, we list out the constants used.
$$c_0=64,\; c_2=1/25,\; c_3=1/16,\; c_4=12 $$

\begin{pfof}{Lemma~\ref{lemma:big-set-expansion}}
	Let $\bs=V(G)\setminus S$.
	The probability that a fixed vertex in $\bs$ does not have edge to $S$ is at most $(1-p)^{|S|}\le(1-p)^{1/p}\le e^{-1}$.
	Since $|\bs|\ge n-2/p\ge n-2n/(c_0\log n)\ge 31n/32$, the expected value of $|N(S)|$ is at least $\left( 31/32 \right)n(1-e^{-1})\ge n/2$.
	Hence, using Chernoff bound, the probability that $|N(S)|<c_2n=n/25$ is at most $e^{-n/8}$.
	Since the number of such $S$ is at most $n^{2/p}= 2^{2n/c_0}\le 2^{n/32}$, we have the lemma by applying union bound.
\end{pfof}
\begin{pfof}{Lemma~\ref{lemma:small-set-expansion}}
	Let $\bs=V(G)\setminus S$.
	Since $|S|\le 1/p\le n/\log n$, we have $|\bs|\ge n/2$ for sufficiently large $n$. 
	Divide $\bs$ into groups of size $\lceil 1/(p|S|)\rceil $.
	The probability that such a group does not have edge to $S$ is at most $(1-p)^{|S|(1/(p|S|))}\le 1/e$.
	The expected number of groups having edge to $S$ is at least $(np|S|/2)(1-1/e)\ge np|S|/4$.	
	Thus, by Chernoff bound, the probability that $|N(S)|\le np|S|/16$ is at most $e^{-np|S|/16}\le 2^{-c_0|S|\log n/16}\le 2^{-4|S|\log n}$.
	The number of sets of size $|S|$ is at most $2^{|S|\log n}$.
	Hence, taking union bound over all $S$ with $|S|\le 1/p$, we get the required lemma. 
\end{pfof}
\begin{pfof}{Lemma~\ref{lemma:halls}}
	First, we prove that for all $C,D\subseteq V(G)$ such that $|C|\ge n/4$,$|D|\ge c_4/p$, and $C\cap D=\emptyset$, there exist at least one edge between $C$ and $D$ with high probability.
	The probability that there is no edge between such a fixed $C$ and $D$ is at most $(1-p)^{(n/4)(c_4/p)}\le e^{-c_4n/4}$.
	The number of pairs of such $C$ and $D$ is at most $2^{2n}$.
	Hence, by taking union bound, the probability that for all $C$ and $D$, the claim holds is at least $1-e^{2n-(c_4n/4)}\ge 1-e^{-n}$.
	
	Using the above claim, we prove that for all $S\subseteq A$, $S$ has at least $|S|-c_4/p$ neighbors in $\ba:=V(G)\setminus A$ with high probability.
	Suppose there is an $S$ which violates the claim. 
	Note that we can assume $|S|\ge c_4/p$, because otherwise the claim is vacuously true.
	Let $B:=\ba\setminus N(S)$.
	There cannot be any edges between $S$ and $B$.
	Also, $|B|\ge \left( n/2 \right)-(|S|-\left( c_4/p \right))$.
	So, $|B|$ is at least $c_4/p$ and when $|B|<n/4$, $|S|$ is at least $n/4$.
	Hence, using the previous claim, there is an edge between $S$ and $B$ with probability at least $1-e^{-n}$.
	Hence, we get a contradiction, and hence our claim is true with probability at least $1-e^{-n}$.
\end{pfof}

\begin{pfof}{Lemma~\ref{lemma:smallcutsize}}
	Let $\calC(S)$ denote $|E(S,V(G)\setminus S)|$.
	For a fixed vertex subset $S$, the expected value of $\calC(S)$ is at most $np|S|$. 
	Therefore, probability that $\calC(S)>np|S|\ge c_0|S|\log n$ is at most $n^{-c_0|S|/2}$ using Chernoff bounds. 
	The probability that $\calC(S)\le np|S|$ for all sets $S$ of size $k$ is at least $1-n^{-c_0k/2+k}\ge 1-n^{-c_0/2+1}$ using union bound and using $k\ge 1$.
	The probability that $\calC(S)\le np|S|$ for all vertex subsets $S$ is at least $1-n^{-c_0/2+2}\ge $ using union bound over all $k\in[n]$. 
\end{pfof}

\end{document}